\documentclass
[preprint,byrevtex,pre,nofootinbib,10pt,balancelastpage,titlepage,twocolumn,nobibnotes]{revtex4}%
\usepackage{amsfonts}
\usepackage{amsmath}
\usepackage{amssymb}
\usepackage{graphicx}%
\providecommand{\U}[1]{\protect\rule{.1in}{.1in}}
\newtheorem{theorem}{Theorem}

\newtheorem{claim}[theorem]{Claim}
\newtheorem{conclusion}[theorem]{Conclusion}

\newtheorem{definition}[theorem]{Definition}

\newtheorem{remark}[theorem]{Remark}

\newenvironment{proof}[1][Proof]{\noindent\textbf{#1.} }{\ \rule{0.5em}{0.5em}}
\ifx\pdfoutput\relax\let\pdfoutput=\undefined\fi
\newcount\msipdfoutput
\ifx\pdfoutput\undefined\else
\ifcase\pdfoutput\else
\ifx\paperwidth\undefined\else
\ifdim\paperheight=0pt\relax\else\pdfpageheight\paperheight\fi
\ifdim\paperwidth=0pt\relax\else\pdfpagewidth\paperwidth\fi
\fi\fi\fi
\begin{document}
\preprint{UATP/1805}
\title{Consequences of the Detailed Balance for the Crooks Fluctuation Theorem}
\author{P.D. Gujrati}
\email{pdg@uakron.edu}
\affiliation{Department of Physics, Department of Polymer Science, The University of Akron,
Akron, OH 44325}

\begin{abstract}
We show that the assumptions of the detailed balance and of the initial
equilibrium macrostate, which are central to the Crooks fluctuation theorem
(CFT), lead to all microstates along a trajectory to have \emph{equilibrium}
probabilities. We also point out that the Crooks's definition of the backward
trajectory does \emph{not} return the system back to its initial microstate.
Once corrected, the detailed balance assumption makes the CFT a theorem only
about reversible processes involving reversible trajectories that satisfy
Kolmogorov's criterion. As there is no dissipation, the CFT cannot cover
irreversible processes, which is contrary to the common belief. This is
consistent with our recent result that the JE is also a result only for
reversible processes.

\end{abstract}
\date{September 6, 2018}
\maketitle

\section{Introduction}

\subsection{Background}

Crooks fluctuation theorem (CFT) \cite{Crooks,Crooks-PRE-2000} has become a
central piece in the development of a microscopic theory of nonequilibrium
thermodynamics (NEQT) expressed in terms of exchange quantities (such as
exchange work $\Delta_{\text{e}}W$ or exchange heat $\Delta_{\text{e}}Q$
having the suffix e) with the medium $\widetilde{\Sigma}$

, which is external to the system $\Sigma$ but interacts with it. The exchange
quantities are controlled by the \emph{medium-intrinsic} (MI) properties of
the medium and are easy to handle as the medium is normally taken to be in
equilibrium with no irreversibility \cite{Notation}; the irreversibilty if
present is always associated with the system $\Sigma$. This has made their
usage very convenient so the CFT has attracted a lot of interest over the past
two decades or so. We denote this particular thermodynamics here by
$\mathring{\mu}$NEQT, with $\mu$ referring to its microscopic nature and the
dot over it referring to the fact that it is based on extending the concept of
macroscopic exchange quantities \cite{Note-SI,Note2,deGroot,Prigogine,Landau}
to the microscopic level. This version should be contrasted with another
version, to be denoted by $\mu$NEQT here, which is based on
\emph{system-intrinsic} (SI) quantities like the generalized work $\Delta W$
and heat $\Delta Q$; see Sec. \ref{Sec-mNEQT} for explanation. Because the
conventional macroscopic NEQT also uses the exchange quantities, the
$\mathring{\mu}$NEQT is supposed to be a direct extension of it
\cite{deGroot,Prigogine,Landau} to the microscopic level but which does not
directly involve any \emph{force imbalance} between the external force and the
internally induced restoring force necessary for irreversibility \cite{Note2};
however, the extension provided by the $\mu$NEQT incorporates explicitly the
irreversibility of the process
\cite{Note2,Gujrati-GeneralizedWork,Gujrati-GeneralizedWork-Expanded} by
directly involving the force imbalance. Within the $\mathring{\mu}$NEQT, the
CFT provides a detailed stochastic description of work fluctuations to the
situation where the final macrostate need not be in equilibrium (EQ). It also
provides another proof of the Jarzynski identity (JE) \cite{Jarzynski}. The JE
applies to the case where external driving is carried out in one shot, while
it is applied several times during short durations in the approach taken by
Crooks (the details are given later in Sec. \ref{Sec-CrooksApproach}).

\begin{definition}
\label{Def-Microquantity}An extensive quantity defined for a microstate or a
trajectory will be called a \emph{microquantity} in this work, such as
microentropy, microwork, microheat, etc. A thermodynamic average quantity will
be called a \emph{macroquantity}.
\end{definition}

\begin{definition}
\label{Def-Trajectory}The time-evolution of a microstate \textsf{a} in a
thermodynamic process is called the trajectory followed by \textsf{a} during
the process. A deterministic or stochastic trajectory followed by a single
microstate \textsf{a} will be denoted by $\gamma_{\mathsf{a}}$ or simply
$\gamma$.
\end{definition}

\begin{definition}
\label{Def-Trajectories} A trajectory involving microstate transitions
\textsf{a}$\mathfrak{\rightarrow}$\textsf{a}$^{\prime}$ between distinct
microstates will be denoted by $\overline{\gamma}$ and called a \emph{mixed} trajectory.
\end{definition}

The difference between a deterministic and a stochastic $\gamma$ is that the
probability of the microstate does not change in the former but changes in the
latter trajectory. However, the microstate does not change. In the mixed
trajectory, both the microstates and the probabilities change so
$\overline{\gamma}$\ is always stochastic and so is the transition
\textsf{a}$\mathfrak{\rightarrow}$\textsf{a}$^{\prime}$.

Crooks consider a process involving mixed trajectories $\overline{\gamma}$.
The main idea behind Crooks' approach is to identify a first-law-like
statement for a trajectory $\overline{\gamma}_{\text{\textsf{ab}}}$ between
two macrostates \textsf{A} and \textsf{B }in a process $\mathcal{P}$, where
\textsf{a }and \textsf{b} refer to some microstates of \textsf{A} and
\textsf{B}, respectively. The macrostates refer to a collection of $r$
microstates in the set $\left\{  \mathsf{a}_{k}\right\}  ,k=1,\cdots,r$ and
their energy and probability sets $\mathbf{E}_{\text{\textsf{A}}}%
\doteq\left\{  E_{\text{\textsf{A}}k}\right\}  ,\mathbf{p}_{\text{\textsf{A}}%
}\doteq\left\{  p_{\text{\textsf{A}}k}\right\}  $ and $\mathbf{E}%
_{\text{\textsf{B}}}\doteq\left\{  E_{\text{\textsf{B}}k}\right\}
,\mathbf{p}_{\text{\textsf{B}}}\doteq\left\{  p_{\text{\textsf{B}}k}\right\}
$, respectively \cite{Note-microstates}. According to Crooks, the
\emph{extension} of the macroscopic first law (see Eq. (\ref{FirstLaw-MI})
below) for the trajectory $\overline{\gamma}_{\text{C,\textsf{ab}}}$\ is
\begin{equation}
\Delta E(\overline{\gamma}_{\text{C,\textsf{ab}}})=\Delta_{\text{e}%
}Q_{\text{C}}(\overline{\gamma}_{\text{C,\textsf{ab}}})+\Delta_{\text{e}%
}R_{\text{C}}(\overline{\gamma}_{\text{C,\textsf{ab}}});
\label{Trajectory-First-Law}%
\end{equation}
here, $\Delta_{\text{e}}Q_{\text{C}}(\overline{\gamma}_{\text{C,\textsf{ab}}%
})$ denotes the microheat exchanged with the medium $\widetilde{\Sigma}$ that
is added to the system and $\Delta_{\text{e}}R_{\text{C}}(\overline{\gamma
}_{\text{C,\textsf{ab}}})$ the exchange microwork done on the system $\Sigma$
by $\widetilde{\Sigma}$ \cite{Note-SI} over the trajectory $\overline{\gamma
}_{\text{C,\textsf{ab}}}$; they are properly defined in Eqs.
(\ref{Crooks-Heat-Total}) and (\ref{Crooks-Work-Total}). The suffix C in
various quantities is to emphasize the choice made by Crooks in his approach
\cite{Crooks,Crooks-PRE-2000} described in Sec. \ref{Sec-CrooksApproach}. From
now on, we will use $\Delta R_{\text{C}}$ for $\Delta_{\text{e}}R_{\text{C}}%
$\ as this work is always an exchange work.

For concreteness, we assume the work process $\mathcal{P}$ to change the
extensive work parameter $\lambda(t)$ like the volume $V(t)$ of the system and
take $\Sigma$\ from \textsf{A} to \textsf{B }over a time period\textsf{
}$(0,\tau)$. Let $t_{m}=m\delta t,m=0,1,2,\cdots,n$ denote a sequence of times
over the duration $(0,\tau)$ so that $t_{0}=0$ and $t_{n}=\tau$. We denote
$\lambda(t)$ by $\lambda_{m}$ and the microstate by \textsf{a}$_{k_{m}}$ at
$t=t_{m}$. It changes from $\lambda_{l}$\ to $\lambda_{l+1}$ during the
interval $\delta_{l}\doteq(t_{l},t_{l+1}),l=0,1,2,\cdots,n-1$. During the same
interval, \textsf{a}$_{k_{l}}\rightarrow$\textsf{a}$_{k_{l+1}}$ due to
microheat exchange so that the microstate becomes \textsf{a}$_{k_{l+1}}$ at
$t=t_{l+1}$. The final microstate at the end of $\mathcal{P}$ is
\textsf{a}$_{k_{n}}$. In the following, we will find it convenient to denote
\textsf{a}$_{k}$\ simply by the index $k\in\left(  1,2,\cdots,r\right)  $ so
that the subscript $m$ on it ($k_{m}$) will denote the time $t_{m}$ it appears
in a realization of $\mathcal{P}$; the work parameter associated with $k_{m}$
is $\lambda_{m}$. In the transition $k_{l}\rightarrow k_{l+1}$ during
$\delta_{l}$, $k_{l}$ is the \emph{departing} microstate and $k_{l+1}$ is the
\emph{arriving} microstate.

Each realization of the process involves taking $\Sigma$ from some initial
$k_{0}=$\textsf{a}\ through a sequence of intermediate microstates to the
final microstate $k_{n}=$\textsf{b}. The trajectory $\overline{\gamma
}_{\text{C,\textsf{ab}}}$ refers to a chronological sequence of microstates
$(k_{0},k_{1},k_{2},\cdots,k_{n-1},k_{n})$ that the system $\Sigma$\ passes
through at times $t_{m}$, with $k_{m}\in\left(  1,2,\cdots,r\right)  $. We
will refer to this as the forward trajectory $\overline{\gamma}%
_{\text{C,\textsf{ab}}}^{\text{(F)}}$ or simply as $\overline{\gamma
}_{\text{C}}^{\text{(F)}}$, when the context is clear, and the corresponding
process $\mathcal{P}^{\text{(F)}}$. In the backward trajectory $\overline
{\gamma}_{\text{C,\textsf{ba}}}^{\text{(B)}}$ or simply $\overline{\gamma
}_{\text{C}}^{\text{(B)}}$ for the backward process $\mathcal{P}^{\text{(B)}}$
from \textsf{B }to \textsf{A}, $\Sigma$ passes the microstates through the
reverse sequence $(k_{n}=\mathsf{b},k_{n-1},k_{n-2},\cdots,k_{1},k_{0}%
=$\textsf{a}$)$ so that the initial microstate $k_{0}$ of $\overline{\gamma
}_{\text{C,\textsf{ab}}}^{\text{(F)}}$ is the terminal microstate of
$\overline{\gamma}_{\text{C,\textsf{ba}}}^{\text{(B)}}$. There are $r^{n}$
possible forward and backward trajectories for the process $\mathcal{P}%
^{\text{(F)}}$\ and $\mathcal{P}^{\text{(B)}}$, respectively, with
probabilities $p(\overline{\gamma}_{\text{C,\textsf{ab}}}^{\text{(F)}})$ and
$p(\overline{\gamma}_{\text{C,\textsf{ba}}}^{\text{(B)}})$, respectively. The
collection of all forward or backward trajectories will be denoted by the fold
face $\boldsymbol{\overline{\gamma}}$$_{\text{C}}^{\text{(F)}}$ or
$\boldsymbol{\overline{\gamma}}$$_{\text{C}}^{\text{(B)}}$, respectively. We
will use $\mathcal{P}$\ to denote both processes $\mathcal{P}^{\text{(F)}}%
$\ and $\mathcal{P}^{\text{(B)}}$. Crooks derives the following relation for
each of the forward trajectories and its associated backward trajectory:%
\begin{equation}
e^{\omega_{\text{C}}(\overline{\gamma}_{\text{C,\textsf{ab}}}^{\text{(F)}}%
)}\doteq\frac{p(\overline{\gamma}_{\text{C,\textsf{ab}}}^{\text{(F)}}%
)}{p(\overline{\gamma}_{\text{C,\textsf{ba}}}^{\text{(B)}})}=e^{\Delta
S(\overline{\gamma}_{\text{C,\textsf{ab}}}^{\text{(F)}})-\beta_{0}\Delta
Q_{\text{C}}(\overline{\gamma}_{\text{C,\textsf{ab}}}^{\text{(F)}}%
)},\label{Crooks-FT0}%
\end{equation}
where
\begin{equation}
\Delta S(\overline{\gamma}_{\text{C,\textsf{ab}}}^{\text{(F)}})=\ln
p(\text{\textsf{a}})-\ln p(\text{\textsf{b}})\label{Trajectory entropy diff}%
\end{equation}
is the change in the \emph{microscopic} entropy (\emph{microentropy}
$S($\textsf{a}$)\doteq-\ln p$ for a given microstate \textsf{a}, which occurs
with a probability $p$) of the forward trajectory $\overline{\gamma
}_{\text{C,\textsf{ab}}}^{\text{(F)}}\doteq(k_{0},k_{1},k_{2},\cdots
,k_{n-1},k_{n})$ and $\beta_{0}$ is the inverse temperature of the heat bath.
For a given \textsf{a }and \textsf{b}, $\Delta S(\overline{\gamma
}_{\text{C,\textsf{ab}}}^{\text{(F)}})$ does not depend on various
trajectories in $\boldsymbol{\overline{\gamma}}$$_{\text{C,\textsf{ab}}%
}^{\text{(F)}}$. The exponent
\[
\omega_{\text{C}}(\overline{\gamma}_{\text{C,\textsf{ab}}}^{\text{(F)}}%
)\doteq\Delta S(\overline{\gamma}_{\text{C,\textsf{ab}}}^{\text{(F)}}%
)-\beta_{0}\Delta Q_{\text{C}}(\overline{\gamma}_{\text{C,\textsf{ab}}%
}^{\text{(F)}})
\]
on the right side of Eq. (\ref{Crooks-FT0}) is supposed to denote the
microentropy change $\Delta_{0}S(\overline{\gamma}_{\text{C,\textsf{ab}}%
}^{\text{(F)}})$ or simply $\Delta_{0}S(\overline{\gamma}_{\text{C}%
}^{\text{(F)}})$\ of the \emph{isolated} system $\Sigma_{0}\doteq\Sigma
\cup\widetilde{\Sigma}$ along $\overline{\gamma}_{\text{C,\textsf{ab}}%
}^{\text{(F)}}$, whose thermodynamic average, see Sec.
\ref{Sec-ThermodynamicAverage}, that is, the macroentropy
\begin{equation}
\Delta_{0}S^{\text{(F)}}\doteq\left\langle \Delta_{0}S\right\rangle
_{\boldsymbol{\overline{\gamma}}_{\text{C}}^{\text{(F)}}}%
\label{TotalEntropy-Isolated}%
\end{equation}
over the forward trajectories in the set $\boldsymbol{\overline{\gamma}}%
$$_{\text{C}}^{\text{(F)}}$ gives the \emph{irreversible entropy}
$\Delta_{\text{i}}S^{\text{(F)}}$ generated within the system $\Sigma$, which
satisfies the second law
\begin{equation}
\Delta_{0}S^{\text{(F)}}=\Delta_{\text{i}}S^{\text{(F)}}\geq
0.\label{Irreversible Entropy}%
\end{equation}
For $\Delta_{\text{i}}S^{\text{(F)}}>0$, the CFT represents a NEQ result. For
$\Delta_{\text{i}}S^{\text{(F)}}=0$, the CFT represents an EQ result. Crooks
also emphasizes that the CFT is valid even if \textsf{B }is \emph{not} an
EQ-macrostate. If the final macrostate \textsf{B }denotes an EQ-macrostate,
then one can easily derive a version of the JE.

\begin{remark}
\label{Remark-AvNotation}From the notation $\Phi\doteq\left\langle
\Phi\right\rangle _{\boldsymbol{\overline{\gamma}}_{\text{C}}^{\text{(F)}}}$
indicating the average over $\forall\overline{\gamma}_{\text{C}}^{\text{(F)}%
}\in$$\boldsymbol{\overline{\gamma}}$$_{\text{C}}^{\text{(F)}}$, the function
$\Phi$ inside the angular brackets is a microfunction over the trajectory
$\overline{\gamma}_{\text{C}}^{\text{(F)}}$ and should not be confused with
the thermodynamic average, the macrofunction, on the left side. Thus,
$\Delta_{0}S$ within the angular brackets in $\left\langle \Delta
_{0}S\right\rangle _{\boldsymbol{\overline{\gamma}}_{\text{C}}^{\text{(F)}}}%
$\ stands for the trajectory microentropy change $\Delta_{0}S(\overline
{\gamma}_{\text{C}}^{\text{(F)}})$, and should never be confused with the
macroentropy change $\Delta_{0}S^{\text{(F)}}$, the entropy change of
$\Sigma_{0}$.
\end{remark}

\begin{remark}
\label{Remark-SecondLaw}While $\Delta_{\text{i}}S^{\text{(F)}}\geq0$ is
nonnegative because of the second law, there is no such restriction for
individual trajectory microentropy $\Delta_{0}S(\overline{\gamma}_{\text{C}%
}^{\text{(F)}})$; the latter can be of any sign.
\end{remark}

\subsection{Motivation}

The derivation of Eq. (\ref{Crooks-FT0}) is based on treating the sequence
$(k_{0},k_{1},k_{2},\cdots,k_{n-1},k_{n})$ as a Markov chain satisfying the
principle of \emph{detailed balance}. This is surprising as the principle is
known to only hold for reversible changes \cite{Klein,Klein0}, \textit{i.e.},
the system will show no irreversibility. Therefore, it is important to
understand how the derivation by Crooks overcomes this limitation. It is this
desire that has motivated this work. Moreover, there seems to be an asymmetry
between $\overline{\gamma}_{\text{C}}^{\text{(F)}}$ and $\overline{\gamma
}_{\text{C}}^{\text{(B)}}$ \cite{Crooks}. For $\overline{\gamma}_{\text{C}%
}^{\text{(F)}}$, the\ initial transition $k_{0}\rightarrow k_{1}$ occurs at
\emph{fixed} $\lambda=\lambda_{1}$ associated with the arriving microstate
$k_{1}$ at $t_{1}$, and the final transition $k_{n-1}\rightarrow k_{n}$ occurs
at fixed $\lambda=\lambda_{n}$ associated with the arriving microstate $k_{n}$
at $t_{n}$. Thus, the work parameter for $k_{l}\rightarrow k_{l+1}$ is fixed
at $\lambda=\lambda_{l+1}$ the associated with the arriving microstate at
$t=t_{l+1}$ in $\overline{\gamma}_{\text{C}}^{\text{(F)}}$. However, for
$\overline{\gamma}_{\text{C}}^{\text{(B)}}$, the initial transition
$k_{n}\rightarrow k_{n-1}$ occurs at fixed $\lambda=\lambda_{n}$ associated
with the departing microstate $k_{n}$. Similarly, the final transition
$k_{1}\rightarrow k_{0}$ occurs at fixed $\lambda=\lambda_{1}$ associated with
the departing microstate $k_{1}$. Thus, the work parameter for $k_{l+1}%
\rightarrow k_{l}$ is at fix $\lambda=\lambda_{l+1}$ associated with the
departing microstate $k_{l+1}$ in $\overline{\gamma}_{\text{C}}^{\text{(B)}}$.
This asymmetry appears innocuous but is not as $\overline{\gamma}_{\text{C}%
}^{\text{(B)}}$ does not bring the system back to the equilibrium microstate
$k_{0}$ as will be discussed later; see Sec. \ref{Sec-BackwardTrajectory}.
This requires carefully identifying the correct backward trajectory.

\subsection{New Results}

We always assume $k_{0}=$\textsf{a} to have the equilibrium probability
associated with the equilibrium macrostate \textsf{A}. We show that the
assumption of the principle of detailed balance, when carefully analyzed,
requires the new microstates $\left\{  k_{m}\right\}  _{m=1,\cdots,n}$ to have
equilibrium probabilities \cite{Note-microstates} so that the final macrostate
\textsf{B}\ is also an EQ-macrostate. The fact that the sequence $\left\{
k_{m}\right\}  _{m=0,1,\cdots,n}$ belonging to $\mathcal{P}^{\text{(F)}}$\ is
a sequence of EQ-microstates, see Definition \ref{Def-EQ-NEQ-microstates},
having equilibrium probabilities does not rule out by itself that
$\mathcal{P}^{\text{(F)}}$ is not irreversible as one can imagine an
irreversible process between two equilibrium macrostates. Therefore, we need
some additional argument to establish that $\mathcal{P}^{\text{(F)}}$\ is
reversible. For this we need to consider the backward process. We show that
when we carefully identify the backward process, we find that
\begin{equation}
\omega(\overline{\gamma}_{\text{\textsf{ab}}}^{\text{(F)}})\equiv
\Delta_{\text{i}}S(\overline{\gamma}_{\text{\textsf{ab}}}^{\text{(F)}%
})=0,\forall\overline{\gamma}_{\text{\textsf{ab}}}^{\text{(F)}}
\label{Correct-EntropyGeneration-Trajectory}%
\end{equation}
so that the thermodynamic average over the set $\boldsymbol{\overline{\gamma}%
}$$_{\text{\textsf{ab}}}^{\text{(F)}}$ also vanishes identically
\begin{equation}
\Delta_{\text{i}}S^{\text{(F)}}=0; \label{Frozen-Microentropy}%
\end{equation}
compare with Eq. (\ref{Irreversible Entropy}) and the discussion following it.
Note that the forward trajectory $\overline{\gamma}_{\text{\textsf{ab}}%
}^{\text{(F)}}$ introduced above and the associated backward trajectory
$\overline{\gamma}_{\text{\textsf{ab}}}^{\text{(B)}}$ are not identical to the
Crooks forward and backward processes $\overline{\gamma}_{\text{C,\textsf{ab}%
}}^{\text{(F)}}$ and $\overline{\gamma}_{\text{C,\textsf{ab}}}^{\text{(B)}}$,
respectively, as we will explain later; see Eqs.
(\ref{ForwardTrajectorySequence}) and (\ref{ReverseTrajectorySequence}). This
then completes the demonstration that the CFT, when properly defined, only
covers reversible processes, contrary to what is conventionally accepted in
the field. This is consistent with our previous conclusion
\cite{Gujrati-GeneralizedWork,Gujrati-GeneralizedWork-Expanded} that the JE is
also restricted to reversible processes only.

\subsection{Layout}

The layout of the paper is as follows. In the next section, we briefly review
the two versions of the microscopic NEQT. In particular, we show that the
generalized SI-work \cite{Note-SI} is carried out \emph{isentropically}, while
the generalized SI-heat is always given by $TdS$ no matter how irreversible
the transformation is provided we can identify the temperature $T$. This
ensures that the generalized work ($dS=0$) and generalized heat ($dS\neq0$)
are independent contributions in the first law that could never be confused.
At the microstate or trajectory level, this means that the generalized
microwork is carried out at fixed probabilities ($dp=0$), while the
generalized microheat requires changes in the probabilities ($dp\neq0$). This
distinct separation proves useful for calculating various trajectory
quantities as no confusion can arise between the generalized microwork and
microheat. We summarize Crooks's approach in Sec. \ref{Sec-CrooksApproach}. We
partition the medium $\widetilde{\Sigma}$\ into two separate and mutually
noninteracting parts, the external work source $\widetilde{\Sigma}_{\text{w}}$
and an external heat source $\widetilde{\Sigma}_{\text{h}}$ \cite{Notation}%
.\ Crooks simplifies his approach and allows the system $\Sigma$\ to undergo
interactions with each one at different times; we call these $\widetilde
{\Sigma}_{\text{w}}$-interactions to perform external work and $\widetilde
{\Sigma}_{\text{h}}$-interactions to exchange microheat. The principle of
microscopic detailed balance is treated in Sec. \ref{Sec-DetailedBalance}.
This is an important section where attention is drawn to the fact that the
acceptance of the principle results in the transition matrix elements
$\left\{  T_{ij}\right\}  $ being determined by the equilibrium probabilities
of the arriving state $j$, see Eqs. (\ref{TransitionMatrix-Elements-DB}) and
(\ref{TransitionMatrix-Unique}), and has no memory of the departing state $i$.
The same conclusion follows if we treat the transition matrix to be balanced.
It follows from this conclusion that all microstates after $\widetilde{\Sigma
}_{\text{h}}$-interactions become EQ-microstates (see Definition
\ref{Def-EQ-NEQ-microstates}). We then show in Sec. \ref{Sec-Work-Heat-Order},
see Conclusion \ref{Conclusion-InteractionOrder}, that $\widetilde{\Sigma
}_{\text{w}}$-interactions must precede $\widetilde{\Sigma}_{\text{h}}%
$-interactions during each interval $\delta_{l}$ of the process $\mathcal{P}%
^{\text{(F)}}$ if we require a terminal EQ-microstate at $t=t_{l+1}$; the
order \emph{cannot} be reversed. The following section, Sec.
\ref{Sec-Consequences}, is devoted to the consequences of the detailed
balance. Sec. \ref{Sec-BackwardTrajectory} is also very important, where we
show that the Crooks's backward trajectory does not satisfy the basic
requirement that the reverse process brings the system back to its initial
macrostate. We introduce a novel trick to identify the reverse process in a
transparent manner, which we use to introduce the corrected form of the
backward trajectory that satisfies this basic requirement. If accepted, we
find that the Markov chain becomes \emph{reversible}. We then demonstrate that
$\omega(\overline{\gamma}_{\text{C,\textsf{ab}}}^{\text{(F)}})=0$, see Eq.
(\ref{Correct-EntropyGeneration-Trajectory}), which is simply an extension of
Kolmogorov's Criterion for conventional reversible Markov chain to the present
work where two different interactions are involved. The extended Kolmogorov
criterion shows that there is no irreversible entropy generation
($\Delta_{\text{i}}S^{\text{(F)}}=0$) and finally proves that the corrected
CFT only applies to reversible processes. It cannot apply to irreversible
processes because of the acceptance of the principle of detailed balance. This
conclusion is consistent with our recent result
\cite{Gujrati-GeneralizedWork,Gujrati-GeneralizedWork-Expanded} that the JE is
also restricted to reversible processes. Thus, they both fail to account for
any irreversibility in $\mathcal{P}^{\text{(F)}}$, contrary to the popular belief.

\section{Microscopic NEQ Thermodynamics\label{Sec-mNEQT}}

In this section, we briefly review the two versions: $\mu$NEQT and
$\mathring{\mu}$NEQT and show how they will be used in understanding the
consequences of the CFT. As noted above, $\mathring{\mu}$NEQT exploits
exchange quantities that are determined by the properties of the medium
$\widetilde{\Sigma}$ (see below for more details). Therefore, the
$\mathring{\mu}$NEQT is governed by \emph{medium-intrinsic} (MI) quantities
\cite{Note-SI,Note2,Gujrati-GeneralizedWork,Gujrati-GeneralizedWork-Expanded}
so care must be taken to bring in indirectly the force imbalance necessary for
irreversibility \cite{Note2} in any consideration. In contrast, the $\mu$NEQT
based on SI-quantities already includes irreversible contributions to the
system due to force imbalance; see Sec. \ref{Sec-CrooksApproach} and
\cite{Note2}. Therefore, the two approaches are very different when
irreversibility is present.

\subsection{Thermodynamic Averages\label{Sec-ThermodynamicAverage}}

In general, an equilibrium or nonequilibrium \emph{ensemble average} (EA) is
defined instantaneously, and\ requires identifying (a) the elements
(microstates $\left\{  \mathsf{a}_{k}\right\}  $) of the ensemble and (b)
their instantaneous probabilities $\left\{  p_{k}\right\}  $. The average is
\emph{uniquely} defined over $\left\{  \mathsf{a}_{k}\right\}  $ using
$\left\{  p_{k}\right\}  $\ at each instant, which we identify as the
\emph{instantaneous ensemble average} (IEA). Let $O_{k}$ be some extensive
microquantity pertaining to \textsf{a}$_{k}$ and $dO_{k}$ the change in it
during some infinitesimal process $d\mathcal{P}$. The instantaneous
thermodynamic averages $\left\langle O\right\rangle $ and $\left\langle
dO\right\rangle $\ are defined \cite{Prigogine,Landau} as
\begin{subequations}
\begin{align}
\left\langle O(t)\right\rangle  &  \doteq%
{\textstyle\sum\nolimits_{k}}
O_{k}(t)p_{k}(t),\label{Thermodynamic Average-General}\\
\left\langle dO(t)\right\rangle  &  \doteq%
{\textstyle\sum\nolimits_{k}}
dO_{k}(t)p_{k}(t), \label{Thermodynamic Average-Differential}%
\end{align}
and define the corresponding macroquantities. We will usually not show the
time $t$ unless clarity is needed. In thermodynamics, it is common to simply
use $O$ and $dO$\ for the average also. Therefore, we will be careful to avoid
this simplification if it may cause confusion. The average energy
$E\equiv\left\langle E\right\rangle $ and entropy $S=\left\langle
S\right\rangle $ are such instantaneous average SI-quantities, where $E_{k}$
and $S_{k}=-\ln p_{k}\doteq-\eta_{k}$ denote the microenergy and microentropy
of the microstate \textsf{a}$_{k}$; we have also introduced \emph{Gibbs
probability index} $\eta_{k}=\ln p_{k}$. The infinitesimal thermodynamic
MI-work $dR\equiv\left\langle dR\right\rangle $ done on the system and the
SI-work done by the system $dW\equiv\left\langle dW\right\rangle $ also
represent such an average instantaneous quantity. They involve the microwork
$dR_{k}$ and $dW_{k}$ associated with \textsf{a}$_{k}$ during $d\mathcal{P}$.

We can extend the average notion in Eq.
(\ref{Thermodynamic Average-Differential}) to the trajectory average
$\left\langle \Delta O\right\rangle _{\boldsymbol{\overline{\gamma}%
}_{\text{\textsf{ab}}}}$:%
\begin{equation}
\Delta O_{\boldsymbol{\overline{\gamma}}_{\text{\textsf{ab}}}}\doteq%
{\textstyle\sum\nolimits_{\overline{\gamma}_{\text{\textsf{ab}}}}}
\Delta O(\overline{\gamma}_{\text{\textsf{ab}}})p(\overline{\gamma
}_{\text{\textsf{ab}}}), \label{Thermodynamic Average-Trajectory}%
\end{equation}
in terms of the probabilities $p(\overline{\gamma}_{\text{\textsf{ab}}})$ of
$\overline{\gamma}_{\text{\textsf{ab}}}\in$$\boldsymbol{\overline{\gamma}}%
$$_{\text{\textsf{ab}}}$; here $\Delta O(\overline{\gamma}_{\text{\textsf{ab}%
}})$ denotes the accumulated microvalue of $O$ along $\overline{\gamma
}_{\text{\textsf{ab}}}$
\end{subequations}
\begin{equation}
\Delta O(\overline{\gamma}_{\text{\textsf{ab}}})\doteq\int\nolimits_{\overline
{\gamma}_{\text{\textsf{ab}}}}dO_{k}(t), \label{Accumulation-Trajectory}%
\end{equation}
where $dO_{k}(t)$ refers to the instantaneous microstate \textsf{a}$_{k}$
along $\overline{\gamma}_{\text{\textsf{ab}}}$. All quantities related to the
trajectory are microquantities as opposed to the macroquantities $\left\langle
O\right\rangle $ and $\left\langle dO\right\rangle $\ above. These are the
quantities that are relevant in $\mu$NEQT and $\mathring{\mu}$NEQT.

\subsection{The First Law}

The first law during an infinitesimal process $d\mathcal{P}$ is expressed as a
sum of two SI-contributions \cite{Gujrati-Stat}%
\begin{equation}
d\left\langle E\right\rangle =%
{\textstyle\sum\nolimits_{k}}
E_{k}dp_{k}+%
{\textstyle\sum\nolimits_{k}}
p_{k}dE_{k}. \label{FirstLaw}%
\end{equation}
The first sum represents the \emph{generalized heat}
\begin{subequations}
\begin{equation}
dQ=%
{\textstyle\sum\nolimits_{k}}
E_{k}dp_{k}, \label{Gen-Heat}%
\end{equation}
while the second sum represents $-dW$, the \emph{generalized work}
\cite{Gujrati-GeneralizedWork,Gujrati-Stat,Gujrati-GeneralizedWork-Expanded}%
\begin{equation}
dW=-%
{\textstyle\sum\nolimits_{k}}
p_{k}dE_{k}. \label{Gen-Work}%
\end{equation}
We have shown elsewhere \cite{Gujrati-I,Gujrati-II,Gujrati-Entropy2} that the
generalized heat can be used to turn the conventional Clausius inequality into
a generalized Clausius equality
\end{subequations}
\begin{equation}
dQ=TdS; \label{ClausiusEquality}%
\end{equation}
here $T$ denotes the temperature of the system, which may or may not equal the
temperature $T_{0}$ of the medium. The Clausius equality is consistent with
$dQ$ being an SI-quantity. Thus, the final form of the first law in terms of
the SI-quantities becomes%
\begin{equation}
d\left\langle E\right\rangle =dQ-dW. \label{FirstLaw-SI}%
\end{equation}
It is this formulation of the first law that forms the cornerstone of the
$\mu$NEQT \cite{Gujrati-GeneralizedWork,Gujrati-GeneralizedWork-Expanded} and
will guide us in this work for the simple reason that it provides a
straightforward description of the microstates and trajectories as noted above.

The first law in Eq. (\ref{FirstLaw}) is in terms of the SI-quantities $E_{k}$
and $p_{k}$. However, its important lies in its ability to clearly distinguish
the concept of generalized heat and work. The origin of the generalized heat
$dQ$ is the change in the microstate probabilities, while their microenergies
remain fixed, and the origin of the generalized work $dW$ is the change in the
microenergies, while their probabilities remain fixed. For fixed
probabilities, the entropy $S=-%
{\textstyle\sum\nolimits_{k}}
p_{k}\ln p_{k}$ remains constant. Therefore, the generalized work is the
\emph{isentropic }change in the macroenergy change $d\left\langle
E\right\rangle $, and the change in $d\left\langle E\right\rangle $ due to
$dS$ alone is the generalized heat. This clearly shows the two distinct
sources for $dQ$ and $dW$ and provides a clear distinction between the two
quantities as said above.

The important point to remember is that these SI-microquantities include the
contributions that arise from the mismatch between the system's and medium's
quantity expressed by the force imbalance \cite{Note2}. This is easily seen by
observing that these generalized quantities differ from the exchanged heat and
work $d_{\text{e}}Q$ and $d_{\text{e}}W$, respectively:%
\begin{equation}
dQ=d_{\text{e}}Q+d_{\text{i}}Q,dW=d_{\text{e}}W+d_{\text{i}}%
W,\label{SI-MI-Relations}%
\end{equation}
with $d_{\text{e}}W=-dR$. The differences $d_{\text{i}}Q$ and $d_{\text{i}}W$
are generated within the system and denote irreversible quantities; see Remark
\ref{Reamrk-Internal-Irreversible}. They satisfy an important identity of
their magnitude%
\begin{equation}
d_{\text{i}}Q=d_{\text{i}}W\geq0\label{diQ-diW-equality}%
\end{equation}
but not the source:\ While $d_{\text{i}}Q$ is generated by the internal
changes in the probabilities at fixed $\mathbf{E}=\left\{  E_{k}\right\}  $,
$d_{\text{i}}W$ is generated by the internal changes in the energies at fixed
$\mathbf{p}=\left\{  p_{k}\right\}  $.

The equality in Eq. (\ref{diQ-diW-equality}) ensures that the first law can
also be expressed in terms of the MI-quantities $d_{\text{e}}Q$ and
$d_{\text{e}}W=-dR$:%
\begin{equation}
d\left\langle E\right\rangle =d_{\text{e}}Q-d_{\text{e}}W=d_{\text{e}}Q+dR.
\label{FirstLaw-MI}%
\end{equation}
It is this formulation that is employed in the $\mathring{\mu}$NEQT and in the
CFT as we will discuss below.

\subsection{Microquantities}

Using Eq. (\ref{Thermodynamic Average-Differential}), we can identify
microheat and microwork
\cite{Gujrati-Entropy2,Gujrati-GeneralizedWork-Expanded} from Eq.
(\ref{FirstLaw}). They are
\begin{equation}
dQ_{k}\doteq E_{k}d\eta_{k},dW_{k}\doteq-dE_{k}, \label{Micro-heat-work}%
\end{equation}
that are associated with \textsf{a}$_{k}$. We also note that they have
distinct sources (fixed $E_{k}$ versus fixed $p_{k}$). Let us look at $dW_{k}%
$, which is part of the summand in the second sum in Eq. (\ref{FirstLaw}).
This summand requires no change in the probability. Therefore, $dE_{k}%
=-dW_{k}$ is the \emph{deterministic} change in the energy $E_{k}$ of
\textsf{a}$_{k}$. This explains the deterministic contribution along $\gamma$
in Definition \ref{Def-Trajectories}. The contribution $dQ_{k}$ is part of the
summand in the first sum in Eq. (\ref{FirstLaw}). This summand requires
changes in the probabilities but not in $E_{k}$. This, therefore, represents
the \emph{stochastic} contribution.

We also observe that both $dQ_{k}$ and $dW_{k}$ refer to the microchanges
associated with a single microstate \textsf{a}$_{k}$, except that $dQ_{k}$ is
a stochastic variable as it involves probability change $d\eta_{k}$, while
$dW_{k}$ is a deterministic variable as it involves no probability change.
Therefore, to determine their accumulation using Eq.
(\ref{Accumulation-Trajectory}), we must replace $\overline{\gamma
}_{\text{\textsf{ab}}}$ by $\gamma_{\text{\textsf{a}}}$. This is a tremendous
simplification due to the concept of microheat $dQ_{k}$ in the $\mu$NEQT
compared to the Crooks approach which uses the mixed trajectory $\overline
{\gamma}_{\text{\textsf{ab}}}$, where one is forced to introduce transition
probabilities; see Secs. \ref{Sec-CrooksApproach} and
\ref{Sec-DetailedBalance}.

The thermodynamic averages of microheat and microwork following Eq.
(\ref{Thermodynamic Average-Differential}) are
\begin{equation}
dQ=\left\langle dQ\right\rangle ,dW=\left\langle dW\right\rangle .
\label{Av-Heat-Work}%
\end{equation}
It is these microquantities that lend them useful to study quantities relevant
for a trajectory $\gamma$.

From Eq. (\ref{SI-MI-Relations}), we find that the irreversible work is given
by
\begin{subequations}
\begin{equation}
d_{\text{i}}W=dW+dR\geq0, \label{IrreversibleWork0}%
\end{equation}
where we have used the identification $d_{\text{e}}W=-dR$. We are now ready to
understand the physical significance of $d_{\text{i}}W$. Its microscopic
analog $d_{\text{i}}W_{k}$ for a given microstate \textsf{a}$_{k}$ is the
microwork done by the force imbalance such as $P_{k}-P_{0}$ \cite{Note2} as
shown elsewhere
\cite{Gujrati-GeneralizedWork,Gujrati-GeneralizedWork-Expanded} so that
\begin{equation}
d_{\text{i}}W\equiv%
{\textstyle\sum\nolimits_{k}}
p_{k}d_{\text{i}}W_{k}=\left\langle d_{\text{i}}W\right\rangle .
\label{IrreversibleWork}%
\end{equation}
The microwork $d_{\text{i}}W_{k}$ denotes an internal microwork and not an
irreversible microwork; see Remark \ref{Reamrk-Internal-Irreversible}.

The inequality in Eq. (\ref{diQ-diW-equality}) or in Eq.
(\ref{IrreversibleWork0}) is related to the second law \cite{Gujrati-Entropy2}%
. It is easy to show that
\end{subequations}
\begin{subequations}
\begin{equation}
d_{\text{i}}Q=d_{\text{i}}W=(T-T_{0})d_{\text{e}}S+Td_{\text{i}}S\geq0,
\label{diQ-diS-General}%
\end{equation}
in which each term of the second equation must be nonnegative to satisfy the
second law \cite{Gujrati-GeneralizedWork}. For an isothermal change ($T=T_{0}%
$), we have%
\begin{equation}
d_{\text{i}}Q=d_{\text{i}}W=T_{0}d_{\text{i}}S\geq0.
\label{diQ-diS-Isothermal}%
\end{equation}
The irreversibility ($d_{\text{i}}S>0$) in this case is due to the
irreversibility caused by the performance of work. The term $(T-T_{0}%
)d_{\text{e}}S$ is the irreversibility caused by heat transfer at different
temperature ($T\neq T_{0}$) in addition to the work irreversibility. For a
reversible process, $d_{\text{i}}Q=d_{\text{i}}W=0$. In that case,
$T=T_{0},d_{\text{i}}S=0,dQ=d_{\text{e}}Q=T_{0}dS$ and $dW=d_{\text{e}}W=-dR$.

The microscopic version of Eq. (\ref{diQ-diS-Isothermal}) gives rise to
\end{subequations}
\begin{equation}
d_{\text{i}}W_{k}=d_{\text{i}}Q_{k}. \label{micro-internal-heat-workEquality}%
\end{equation}
It follows from Eq. (\ref{diQ-diS-Isothermal}) that the corresponding
microrelation is $d_{\text{i}}Q_{k}=T_{0}d_{\text{i}}S_{k}$, $d_{\text{i}%
}S\doteq\langle d_{\text{i}}S\rangle$,; compare with Eq.
(\ref{IrreversibleWork}). This microrelation can be used to relate the
microwork $d_{\text{i}}W_{k}$ with the internal microentropy change
$d_{\text{i}}S_{k}$ associated with \textsf{a}$_{k}$ :%
\begin{equation}
d_{\text{i}}W_{k}=T_{0}d_{\text{i}}S_{k},
\label{IsothermaInternallMicrowork-InternalMicroentropy}%
\end{equation}
except that there is no restriction on the sign of $d_{\text{i}}W_{k}$ or
$d_{\text{i}}S_{k}$ at the microstate level.

\begin{remark}
\label{Remark-NoForceImbalance-ReversibleProcess}If a theory does not allow
for any force imbalance for \textsf{a}$_{k}$, $d_{\text{i}}W_{k}$ will be
identically zero for $\forall k$. If the process is isothermal, we see from
Eq. (\ref{IsothermaInternallMicrowork-InternalMicroentropy}) that
$d_{\text{i}}S_{k}\equiv0$. Consequently,
\begin{equation}
d_{\text{i}}S\equiv\left\langle d_{\text{i}}S\right\rangle =0,
\label{NoForceImbalance}%
\end{equation}
which implies that the process is a reversible process. However, as discussed
in Remark \ref{Remark-DR-DH-relation}, there are no fluctuations in
$d_{\text{i}}S_{k}$.
\end{remark}

We will see later how Eqs.
(\ref{IsothermaInternallMicrowork-InternalMicroentropy}) and
(\ref{NoForceImbalance}) are fulfilled in the CFT.

\begin{remark}
\label{Remark-Hindsight}With hindsight, we will only be interested in treating
processes for the CFT that satisfy Eq. (\ref{MI-Work-HamiltonianRelation}).
This means that Crooks does not allow any force imbalance as we discuss below.
\end{remark}

\subsection{Work-Energy Theorem\label{Sec-Work-Energy-Theorem}}

From now on, we will usually refer to the generalized (macro)work and
(macro)heat simply as work and heat for simplicity as the notation will be
explicit and no confusion can arise. The work $dW$ done by the system is an
SI-quantity, which is distinct from the MI-work $dR$ done on the system by the
medium during an infinitesimal process $d\mathcal{P}$. They are related by
\cite{Gujrati-I,Gujrati-II,Gujrati-Entropy2}%
\begin{equation}
dW=-dR+d_{\text{i}}W,d_{\text{i}}W\geq0, \label{dWork-relation}%
\end{equation}
along some infinitesimal process $d\mathcal{P}$; see Eq.
(\ref{IrreversibleWork0}). By introducing the accumulation $\Delta\Psi$ of
some extensive (SI- or MI-) quantity $\Psi$ along a process $\mathcal{P}$
defined by%
\begin{equation}
\Delta\Psi\doteq%
{\textstyle\int\nolimits_{\mathcal{P}}}
d\Psi, \label{Accumulation}%
\end{equation}
the above relationship becomes%
\begin{subequations}
\begin{equation}
\Delta W=-\Delta R+\Delta_{\text{i}}W, \label{DWork-relation}%
\end{equation}
where the meaning of various accumulated quantities follows Eq.
(\ref{Accumulation}).

We now describe how to evaluate microwork $\Delta W(\overline{\gamma
}_{\text{\textsf{ab}}})$ for a mixed trajectory. As the initial and final
microstates can be different, we need to divide $\overline{\gamma
}_{\text{\textsf{ab}}}$ into a disjoint collection of several $\gamma
_{\text{\textsf{a}}}$'s whose union gives $\overline{\gamma}%
_{\text{\textsf{ab}}}$:%
\end{subequations}
\begin{equation}
\overline{\gamma}_{\text{\textsf{ab}}}\equiv\cup_{l}\gamma_{\text{\textsf{a}%
}_{l}}\label{Mix-Pure-TrajectoryCombination}%
\end{equation}
For each $\gamma_{\text{\textsf{a}}_{l}}$, we are dealing with a single
microstate \textsf{a}$_{l}$ so we define%
\begin{subequations}
\begin{equation}
\Delta W(\overline{\gamma}_{\text{\textsf{ab}}})\doteq%
{\textstyle\sum\nolimits_{l}}
\Delta W(\gamma_{\text{\textsf{a}}_{l}}).\label{Mix-Trajectory-SI-Microwork}%
\end{equation}
We similarly define%
\begin{equation}
\Delta R(\overline{\gamma}_{\text{\textsf{ab}}})\doteq%
{\textstyle\sum\nolimits_{l}}
\Delta R(\gamma_{\text{\textsf{a}}_{l}}).\label{Mix-Trajectory-MI-Microwork}%
\end{equation}

In the Hamiltonian formulation used at the microscopic level so that it
becomes relevant for microstates \textsf{a} or their trajectories $\gamma$
along $\mathcal{P}$, we are interested in microquantities and their
fluctuations. For example, we are interested in SI-microquantities such as
$\Delta W(\overline{\gamma}_{\text{\textsf{ab}}})$, the work done by the
system along the trajectory $\overline{\gamma}_{\text{\textsf{ab}}}$, or the
MI-microwork $\Delta R(\overline{\gamma}_{\text{\textsf{ab}}})$ done on the
system; see Eq. (\ref{Accumulation-Trajectory}) for the definition. Their
relationship follows from Eq. (\ref{DWork-relation}):%
\end{subequations}
\begin{equation}
\Delta W(\overline{\gamma}_{\text{\textsf{ab}}})=-\Delta R(\overline{\gamma
}_{\text{\textsf{ab}}})+\Delta_{\text{i}}W(\overline{\gamma}%
_{\text{\textsf{ab}}}), \label{DWork-relation-Trajectory}%
\end{equation}
where $\Delta_{\text{i}}W(\overline{\gamma}_{\text{\textsf{ab}}})$ is the
internal microwork due to force imbalance that must be present for
irreversibility and strongly depends on the nature of the trajectory. However,
there is no restriction on the sign of $\Delta_{\text{i}}W(\overline{\gamma
}_{\text{\textsf{ab}}})$ as is for $\Delta_{\text{i}}W\geq0$.

\begin{remark}
\label{Reamrk-Internal-Irreversible}\textbf{ }We will always call any
$\Delta_{\text{i}}$-microquantity an \emph{internal} microquantity since its
sign is not restricted. The thermodynamic average of such a quantity satisfies
the restriction imposed by the second law and has a particular sign. We can
call such a macroquantity an \emph{irreversible} quantity. We will always make
this distinction. Thus, $\Delta_{\text{i}}W(\overline{\gamma}%
_{\text{\textsf{ab}}})$ is an internal microquantity with no sign restriction,
but the average $\Delta_{\text{i}}W$ satisfies the restriction due to the
second law and is called the irreversible work.
\end{remark}

Let $\Sigma$\ be described by its Hamiltonian $\mathcal{H}(\lambda)$ in which
the work parameter $\lambda$ can be manipulated from the outside through the
medium by applying a "force" $F_{0}$. A particular realization of the system
corresponds to a microstate \textsf{a}$_{k}$ of the system so the discussion
below should be taken for at the microstate level. We denote the value of the
Hamiltonian for \textsf{a}$_{k}$ by $E_{k}$. The restoring force in the system
conjugate to $\lambda$ is given by the SI-force%
\begin{equation}
F_{k}=-\partial E_{k}(\lambda)/\partial\lambda,\label{SI-Force}%
\end{equation}
which need not be equal to $F_{0}$ in magnitude. The SI-microwork
$dW_{k}=F_{k}d\lambda$ is given by
\begin{equation}
dW_{k}=-d_{\lambda}E_{k}(\lambda),\label{SI-Work}%
\end{equation}
the negative of the change $d_{\lambda}E_{k}(\lambda)$ in the energy
$E_{k}(\lambda)=\mathcal{H}_{k}(\lambda)$ due to the variation of $\lambda$ as
exemplified by the derivative operator $d_{\lambda}$. We consider a trajectory
$\gamma_{k}$ associated with \textsf{a}$_{k}$ during which $\lambda$ changes
from its initial value $\lambda_{\text{in}}$ to its final value $\lambda
_{\text{fn}}$. The above definition of the SI-force $F_{k}$ and the microwork
$dW_{k}$ above justifies Eq. (\ref{SI-Work-HamiltonianRelation}), which
appears as an identity due to the definition of the force $F_{k}$.

\begin{theorem}
\label{Theorem-Work-Energy-Theorem}It is the SI-microwork that is directly
related to the change in the Hamiltonian $\mathcal{H}_{k}(\lambda)$, known as
the \emph{Work-Energy Theorem }%
\cite{Gujrati-GeneralizedWork,Gujrati-GeneralizedWork-Expanded}:%
\begin{equation}
\Delta W(\gamma_{k})\equiv-\Delta_{\lambda}\mathcal{H}_{k}(\lambda
)\doteq\mathcal{H}_{k}(\lambda_{\text{in}})-\mathcal{H}_{k}(\lambda
_{\text{fn}}), \label{SI-Work-HamiltonianRelation}%
\end{equation}
whether internal work $\Delta_{\text{i}}W_{k}$ is present or not.
\end{theorem}

\begin{proof}
The theorem follows trivially and directly from the definition of the SI-force
and work; see Eqs. (\ref{SI-Force}) and (\ref{SI-Work}) for $F_{k}$ and
$dW_{k}$, respectively. Accumulating $dW_{k}$ along $\gamma_{k}$ immediately
proves Eq. (\ref{SI-Work-HamiltonianRelation}) as an identity.
\end{proof}

The proof does not require any relation between $F_{k}$ and $F_{0}$ so force
imbalance can exist, which will result in a nonzero internal microwork
$\Delta_{\text{i}}W(\gamma_{k})$.

\begin{remark}
\label{Remark-DR-DH-relation}It is clear that%
\begin{equation}
\Delta R(\gamma_{k})=\mathcal{H}_{k}(\lambda_{\text{fn}})-\mathcal{H}%
_{k}(\lambda_{\text{in}}) \label{MI-Work-HamiltonianRelation}%
\end{equation}
only if $\Delta_{\text{i}}W(\gamma_{k})\equiv0$. While $\Delta_{\text{i}%
}W(\gamma_{k})\equiv0$ implies $\Delta_{\text{i}}W=0$, the converse is not
always true as it is possible to have several nonzero $\Delta_{\text{i}%
}W(\gamma_{k})$ and yet have their thermodynamic average $\Delta_{\text{i}%
}W=0$. In the former case, there is no fluctuation in $\Delta_{\text{i}%
}W(\gamma_{k})$, while the fluctuations are present in the latter case.
Therefore, the validity of Eq. (\ref{MI-Work-HamiltonianRelation}) requires a
very strong statement about the absence of fluctuations
\cite{Gujrati-GeneralizedWork,Gujrati-GeneralizedWork-Expanded}.
\end{remark}

We clearly see that the $\mu$NEQT is more appropriate than the $\mathring{\mu
}$NEQT to deal with microscopic fluctuations and force imbalance.

\section{Crooks' Approach\label{Sec-CrooksApproach}\textsc{\qquad}}

We now review and elaborate the approach taken by Crooks. Let us consider a
finite segment $\delta\mathcal{P}$ of the process $\mathcal{P}$ and the
portion $\delta\overline{\gamma}_{\text{\textsf{a}}^{\prime}\text{\textsf{a}%
}^{\prime\prime}}$ of $\overline{\gamma}_{\text{\textsf{ab}}}$\ associated
with $\delta\mathcal{P}$. Crooks identifies the change in the energy of the
microstate \textsf{a}$^{\prime}$ along $\delta\overline{\gamma}%
_{\text{\textsf{a}}^{\prime}\text{\textsf{a}}^{\prime\prime}}$ by the
performance of microwork $\Delta R_{\text{C}}$ by $\Sigma_{\text{w}}$ on the
microstate, see Eq. (\ref{Crooks-Work}), as the microwork parameter
changes$\ (\lambda\rightarrow\lambda^{\prime})$ for the system; the microstate
itself and its probability do not change. During the exchange of energy
$\Delta_{\text{e}}Q_{\text{C}}$ with $\Sigma_{\text{h}}$, see Eq.
(\ref{Crooks-Heat}), the exchange causes microstate and its probability to
change but $\lambda^{\prime}$ does not change. In general, the instantaneous
energy of a microstate \textsf{a}$_{k}$ is determined by the current value of
$\lambda$ at that instant and we denote it by either $E_{k}(\lambda)~$or
$E(k,\lambda)$, depending on which form is more suitable in the context. If
both $E(\lambda)$ and $\lambda$ appear as arguments in a quantity, we will
suppress $\lambda$ in $E(\lambda)$ and use $E,\lambda$ as the argument. The
energy change along $\delta\overline{\gamma}$ is determined by using the
first-law extension for the trajectory given in Eq.
(\ref{Trajectory-First-Law}).

Crooks considers the following protocol for the change $\Delta E_{\overline
{\gamma}}$ over the duration $(0,\tau)$ by introducing contiguous
nonoverlapping $n$ time intervals $\delta_{l}\doteq(t_{l},t_{l+1})$ associated
with $\delta\mathcal{P}_{l},l=0,1,2,\cdots,n-1$, each of which is further
divided into two nonoverlapping parts: an \emph{earlier} and a \emph{latter}
part $\delta_{l}^{\prime}\doteq(t_{l},t_{l+1}^{\prime})$ and \ $\delta
_{l}^{\prime\prime}\doteq(t_{l+1}^{\prime},t_{l+1}),t_{l+1}^{\prime}<t_{l+1}$,
respectively, with $\delta_{l}^{\prime}\cup\delta_{l}^{\prime\prime}%
\doteq\delta_{l}$. (We will establish below in Sec. \ref{Sec-Consequences}
that $\tau$ must be taken to be $\tau_{\text{eq}}$ as the final macrostate in
the Crooks process turns out to be an EQ-macrostate, contrary to what is
commonly believed.) During $\delta_{l}^{\prime}$, only microwork is exchanged
and during $\delta_{l}^{\prime\prime}$, only microheat is exchanged. The
trajectory $\overline{\gamma}_{\text{C}}\equiv\overline{\gamma}_{\text{C}%
\mathsf{ab}}$ is uniquely specified by the chronologically ordered microstate
sequence $\{k_{0}\equiv\mathsf{a},k_{1},k_{2},\cdots,k_{n-1},k_{n}%
\equiv\mathsf{b}\}$, in which $k_{m},m=0,1,2,\cdots,n$ refers to the
microstate \textsf{a}$_{k_{m}}(t_{m})$ at time $t_{m}$. The time $t_{m}$ also
specifies the value of $\lambda=\lambda_{m}$\ and the corresponding
probability $p_{k_{m}}(t_{m})$; here $k_{m}\doteq k\left(  t_{m}\right)
\in\left\{  1,2,\cdots,r\right\}  $ indexes the $r$\ microstates. Crooks
further assumes that the sequence $\left\{  k_{m}\right\}  $ (along with the
associated sequence $\left\{  p_{k_{m}}(t_{m})\right\}  $
\cite{Note-microstates}) forms a (time-inhomogeneous) Markov chain, which is
divided into $n$ contiguous segments $\left\{  k_{l},k_{l+1}\right\}  $ in
one-to-one correspondence with the $n$ intervals $\delta_{l}$\ so that the
integration over $(0,\tau)$ is \emph{approximated} by a sum over these
intervals, its accuracy getting better as $n$ increases. Each microstate is
specified by its energy and work parameter at time $t_{m}$. The probability
$p_{k_{m}}$ of a microstate $k_{m}$\ may or may not have its equilibrium value.

\begin{definition}
\label{Def-EQ-NEQ-microstates}We will call a microstate (by an abuse of the
concept only valid for macrostates) an EQ-microstate or a NEQ-microstate if
the corresponding probability is the equilibrium probability or not,
respectively; the latter will be denoted by appending a prime on the
microstate label such as $k^{\prime}$.
\end{definition}

Crooks does not make any such distinction but we will find it very useful in
our discussion and also when identifying the backward trajectories. During
$\delta_{l}^{\prime}$ over which $\Sigma$ interacts with $\Sigma_{\text{w}}$,
$\lambda_{l}$ changes to $\lambda_{l+1}$ but not the microstate $k_{l}$.
During this interval,\ microwork done \emph{on} the system is taken to be the
energy change of the microstate:
\begin{subequations}
\begin{equation}
\Delta R_{\text{C}}^{(l)}(k_{l},\lambda_{l}\rightarrow\lambda_{l+1})\doteq
E(k_{l}^{\prime},\lambda_{l+1})-E(k_{l},\lambda_{l})\label{Crooks-Work}%
\end{equation}
with no change in the probability $p(k_{l},\lambda_{l})$. This is inconsistent
with the work-energy theorem according to which%
\begin{equation}
\Delta W^{(l)}(k_{l},\lambda_{l}\rightarrow\lambda_{l+1})\doteq E(k_{l}%
,\lambda_{l})-E(k_{l}^{\prime},\lambda_{l+1}),\label{Gujrati-Work}%
\end{equation}
as follows from Eq. (\ref{SI-Work-HamiltonianRelation}). We have also used the
notation $k_{l}^{\prime}$ above to denote the microstate $k_{l}\left(
t_{l+1}^{\prime}\right)  $ that results at the end of the work protocol at
$t_{l+1}^{\prime}$. It represents a NEQ-microstate as will become clear in
Sec. \ref{Sec-DetailedBalance}; see also Conclusion
\ref{Conclusion-Microstate-Nature}. The only difference between $k_{l}$ and
$k_{l}^{\prime}$ is in their probabilities. At present, this difference is not
relevant so we can overlook the prime for the moment. The work $\Delta
_{\text{e}}R_{\text{C}}(\overline{\gamma}_{\text{C,\textsf{ab}}})$ in Eq.
(\ref{Trajectory-First-Law}) is given by%
\begin{equation}
\Delta R_{\text{C}}(\overline{\gamma}_{\text{C,\textsf{ab}}})=%
{\textstyle\sum\nolimits_{l}}
\Delta R_{\text{C}}^{(l)}(k_{l},\lambda_{l}\rightarrow\lambda_{l+1}%
),\label{Crooks-Work-Total}%
\end{equation}
with fixed probability set $\left\{  p(k_{l},\lambda_{l})\right\}
_{l=0,1,\cdots,n-1}$; compare with Eq. (\ref{Mix-Trajectory-MI-Microwork}).
Comparing $\Delta R_{\text{C}}^{(l)}$ with the energy change in Eq.
(\ref{SI-Work}), we conclude that Crooks has implicitly assumed that Eq.
(\ref{MI-Work-HamiltonianRelation}) remains valid so no force imbalance is
allowed; see Remark \ref{Remark-DR-DH-relation}. The work parameter
$\lambda(t)=\lambda_{l+1}$\ is held fixed over $\delta_{l}^{\prime\prime}$
during which the microstate changes from $k_{l}$ to $k_{l+1}$ by exchanging
energy with $\Sigma_{\text{h}}$ in the form of microheat
\end{subequations}
\begin{subequations}
\begin{equation}
\Delta_{\text{e}}Q_{\text{C}}(k_{l}\rightarrow k_{l+1},\lambda_{l+1}%
)=E(k_{l+1},\lambda_{l+1})-E(k_{l},\lambda_{l+1});\label{Crooks-Heat}%
\end{equation}
this exchange occurs in conjunction with probability change as will be
discussed in Sec. \ref{Sec-DetailedBalance}; the latter is however not
involved in $\Delta_{\text{e}}Q_{\text{C}}(k_{l}\rightarrow k_{l+1}%
,\lambda_{l+1})$. The exchange heat $\Delta_{\text{e}}Q_{\text{C}}%
(\overline{\gamma}_{\text{C,\textsf{ab}}})$ in Eq. (\ref{Trajectory-First-Law}%
) is given by%
\begin{equation}
\Delta_{\text{e}}Q_{\text{C}}(\overline{\gamma}_{\text{C,\textsf{ab}}})=%
{\textstyle\sum\nolimits_{l}}
\Delta Q_{\text{C}}^{(l)}(k_{l}\rightarrow k_{l+1},\lambda_{l+1}%
).\label{Crooks-Heat-Total}%
\end{equation}

The forward trajectory is specified uniquely by the following sequence of
microstates, where we reinsert the prime on the microstates as necessary:%
\end{subequations}
\begin{align}
\overline{\gamma}^{\text{(F)}} &  :\overbrace{k_{0}\overset{\text{w}%
}{\rightarrow}k_{0}^{\prime}\overset{\text{h}}{\rightarrow}k_{1}}%
\underbrace{\overset{\text{w}}{\rightarrow}k_{1}^{\prime}\overset{\text{h}%
}{\rightarrow}k_{2}}\overbrace{\overset{\text{w}}{\rightarrow}k_{2}^{\prime
}\overset{\text{h}}{\rightarrow}k_{3}}\cdots\nonumber\\
&  \cdots\underbrace{k_{n-1}\overset{\text{w}}{\rightarrow}k_{n-1}^{\prime
}\overset{\text{h}}{\rightarrow}k_{n}},\label{Forward-trajectory}%
\end{align}
in which the triplet $k_{l}\overset{\text{w}}{\rightarrow}k_{l}^{\prime
}\overset{\text{h}}{\rightarrow}k_{l+1}$ corresponds to the interval
$\delta_{l}$: the w-arrow $k_{l}\overset{\text{w}}{\rightarrow}k_{l}^{\prime}$
refers to the deterministic interaction with $\Sigma_{\text{w}}$ during
$\delta_{l}^{\prime}$ and the h-arrow $k_{l}^{\prime}\overset{\text{h}%
}{\rightarrow}k_{l+1}$to the stochastic interaction with $\Sigma_{\text{h}}$
during $\delta_{l}^{\prime\prime}$. We remark that in each interval
$\delta_{l}$, microwork is performed before microheat is exchanged. These two
interactions are similar to the driven and reequilibration stages used in the
derivation of the JE \cite{Jarzynski}. Thus, the Crooks process is a sequence
of $n$ different Jarzynski processes $\delta\mathcal{P}_{l}$ occurring during
$\delta_{l};l=0,1,\cdots,n-1$. As we will see in Sec.
\ref{Sec-DetailedBalance}, the driven stage drives an EQ-microstate $k_{l}$ to
a NEQ-microstate $k_{l}^{\prime}$, which is then brought to an EQ-microstate
$k_{l+1}$ during the equilibration stage, which means that each Jarzynski
process $\delta\mathcal{P}_{l}$ takes an EQ-microstate to an EQ-microstate. We
summarize this as a

\begin{conclusion}
The w-arrow $\overset{\text{w}}{\rightarrow}$ always takes an EQ-microstate to
a NEQ-microstate, whereas the h-arrow $\overset{\text{h}}{\rightarrow}%
$\ always takes a NEQ-microstate to an EQ-microstate in a forward trajectory;
the two arrows denote the interactions with $\Sigma_{\text{w}}$ and
$\Sigma_{\text{h}}$, respectively.
\end{conclusion}

We also note that because of the presence of intermediate NEQ-microstates
$k_{l}^{\prime}$, $\overline{\gamma}^{\text{(F)}}$ contains more information
than the original Crooks trajectory $\overline{\gamma}_{\text{C}}^{\text{(F)}%
}$.

\section{Microscopic Detailed Balance\label{Sec-DetailedBalance}}

As introduced above, the microwork during $\delta_{l}^{\prime}$ is purely
mechanical, and is carried out at fixed $p_{k_{l}}(E_{l},\lambda_{l},t_{l})$.
Even if $k_{l}$ is an EQ-macrostate at $t_{l}$ [$p_{k_{l}}(E_{l}(\lambda
_{l}),\lambda_{l},t_{l})=p_{k_{l}\text{eq}}(E_{l},\lambda_{l})$], $k_{l}$ at
$t_{l+1}^{\prime}\ $has $p_{k_{l}}(E_{l}^{\prime}(\lambda_{l+1}),\lambda
_{l+1},t_{l+1}^{\prime})=p_{k_{l}\text{eq}}(E_{l}(\lambda_{l}),\lambda_{l})$
as its probability, which is not the equilibrium probability $p_{k_{l}%
\text{eq}}(E_{l}^{\prime}(\lambda_{l+1}),\lambda_{l+1})$ at $\lambda
(t_{l+1}^{\prime})=\lambda_{l+1}$. Therefore, the microstate $k_{l}$ at
$t_{l+1}^{\prime}$ after microwork has been performed is a NEQ-microstate, and
should be denoted with a prime following our convention; See Definition
\ref{Def-EQ-NEQ-microstates}. This explains the reason for introducing
$k_{l}^{\prime}$ in Eq. (\ref{Forward-trajectory}).

The microheat transfer over $\delta_{l}^{\prime\prime}$ is accompanied by
probability changes as we now discuss using the Markovian property encoded in
the transition matrix $\mathbf{T}^{(l)}(\delta_{l}^{\prime\prime})$. The
transition matrix tells us how a NEQ-microstate $k_{l}^{\prime}$ transforms
into $k_{l+1}$. Suppressing $l$ for the moment, we introduce its matrix
elements, the one-step transition probabilities $T_{ij}(\delta^{\prime\prime
})\equiv T(\left.  j\right\vert i\mid\delta^{\prime\prime})$ from microstate
$i^{\prime}\equiv i(t^{\prime})$ to a microstate $j(t^{\prime}+\delta
^{\prime\prime})$ at given $E_{j}^{\prime}$ and $\lambda^{\prime}$; here,
$E_{j}^{\prime}$ and $\lambda^{\prime}$ refer to the microstate $i^{\prime}$
at time $t^{\prime}$. Recall that $\lambda^{\prime}$ is held fixed during
$\delta^{\prime\prime}$ so $\lambda^{\prime}\doteq\lambda(t^{\prime}%
)\equiv\lambda(t^{\prime}+\delta^{\prime\prime})$. Similarly, $E_{j}^{\prime
}\doteq E_{j}(t^{\prime})\equiv E_{j}(t^{\prime}+\delta^{\prime\prime})$. With
hindsight, we are using $j$ for the arriving microstate as we will see that it
an EQ-microstate. We are also suppressing the prime on $i$ in the subscript of
$T_{ij}(\delta^{\prime\prime})$. The matrix elements satisfy%
\begin{equation}%
{\textstyle\sum\nolimits_{j}}
T_{ij}(\delta^{\prime\prime})=1,\label{T-sumrule}%
\end{equation}
and determine the probabilities at the next time $t^{\prime\prime}\doteq
t^{\prime}+\delta^{\prime\prime}$ in the sequence%
\begin{equation}
p_{j}(t^{\prime\prime})=%
{\textstyle\sum\nolimits_{i}}
p_{i}(t^{\prime})T_{ij}(\delta^{\prime\prime}).\label{MarkovProperty1}%
\end{equation}
Given $\mathbf{T}(\delta^{\prime\prime})$, we can determine the new
probability at $t^{\prime\prime}$ in terms of the set $\left\{  p_{i}%
(t^{\prime})\right\}  $, where $p_{i}(t^{\prime})$ stands for $p_{i}%
(E_{i}^{\prime},\lambda^{\prime},t^{\prime})=p_{i}(E_{i}^{\prime}%
(\lambda^{\prime}),\lambda^{\prime},t^{\prime})$ and denotes the probability
of $i^{\prime}$ at given $E_{i}^{\prime}(\lambda^{\prime})=E_{i}(t^{\prime})$
and $\lambda^{\prime}=\lambda(t^{\prime})$. If we introduce a row probability
vector $\mathbf{p}(t^{\prime})\doteq\left\{  p_{i}(t^{\prime})\right\}  $, we
can express the above relation using matrix multiplication%
\begin{equation}
\mathbf{p}(t^{\prime\prime})=\mathbf{p}(t^{\prime})\mathbf{T}(\delta
^{\prime\prime})\label{MarkovProperty2}%
\end{equation}
for the given $\mathbf{E}^{\prime}(\lambda^{\prime})\doteq\{E_{i}^{\prime
}(\lambda^{\prime})\}$ and $\lambda^{\prime}$. It is evident that
$\mathbf{T}(\delta^{\prime\prime})$ depends on $\delta^{\prime\prime}$ in some
fashion. The dependence is not important to know for the discussion here. From
these conditions, we find that%
\begin{equation}
p_{j}(t^{\prime\prime})-p_{j}(t^{\prime})=%
{\textstyle\sum\nolimits_{i}}
\left[  p_{i}(t^{\prime})T_{ij}(\delta^{\prime\prime})-p_{j}(t^{\prime}%
)T_{ji}(\delta^{\prime\prime})\right]  .\label{ProbabilityChange}%
\end{equation}
It is clear that in the Markovian approximation, we need to know
$T_{ij}(\delta^{\prime\prime})$ to determine the probability change. To make
further progress, we need to make some assumption. It is convenient and very
common in the field to determine $T_{ij}(\delta^{\prime\prime})$ by accepting
the \emph{condition of microscopic reversibility}, also known as the condition
of \emph{the detailed balance}, which is valid for equilibrium probabilities
at fixed $\lambda^{\prime}$ at $t^{\prime\prime}$:%
\begin{equation}
p_{i\text{eq}}(E_{i}^{\prime},\lambda^{\prime})\overleftrightarrow{T}%
_{ij}(\delta^{\prime\prime})-p_{j\text{eq}}(E_{j}^{\prime},\lambda^{\prime
})\overleftrightarrow{T}_{ji}(\delta^{\prime\prime})=0,\forall
(i,j),\label{DetailedBalance}%
\end{equation}
where we have used the double arrow to indicate the condition of microscopic
reversibility, and have used the fact that the equilibrium probabilities do
not depend on time. We remark that $p_{i\text{eq}}(E_{i}^{\prime}%
,\lambda^{\prime})$ at $t^{\prime\prime}$ must not be confused with
$p_{i}(E_{i}^{\prime},\lambda^{\prime},t^{\prime})$ at $t^{\prime}$. So that
there cannot be any confusion, we write $p_{i\text{eq}}(E_{i}^{\prime}%
,\lambda^{\prime})$ as $p_{i\text{eq}}(E_{i},\lambda^{\prime})$ since
$E_{i}^{\prime}(\lambda^{\prime})$ at $t^{\prime}$ and $E_{i}(\lambda^{\prime
})$ at $t^{\prime\prime}$ are the same. The use of Eq. (\ref{DetailedBalance})
now determines the elements $\overleftrightarrow{T}_{ij}(\delta^{\prime\prime
})$ \cite{Crooks} as follows. We first conclude that%
\[
\overleftrightarrow{T}_{ij}(\delta^{\prime\prime})/\overleftrightarrow{T}%
_{ji}(\delta^{\prime\prime})=p_{j\text{eq}}(E_{j},\lambda^{\prime
})/p_{i\text{eq}}(E_{i},\lambda^{\prime}).
\]
It then follows from this that the choice
\begin{equation}
\overleftrightarrow{T}_{ij}(\delta^{\prime\prime})=p_{j\text{eq}}%
(E_{j},\lambda^{\prime}),\overleftrightarrow{T}_{ji}(\delta^{\prime\prime
})=p_{i\text{eq}}(E_{i},\lambda^{\prime}),\label{TransitionMatrix-Elements-DB}%
\end{equation}
is consistent with the condition of the detailed balance. A direct proof comes
from Theorem \ref{Theorem-LimitTheorem}. We can summarize the above as the
following claim:

\begin{claim}
\label{Claim-Time-independence of T}The matrix element $T_{ij}(\delta
^{\prime\prime})$ for the transition $i^{\prime}\rightarrow j$ is equal to the
equilibrium probability $p_{j\text{eq}}(E_{j},\lambda^{\prime})$,
i.e.,$\ p_{j\text{eq}}(E_{j}(\lambda^{\prime}),\lambda^{\prime})$ of the
arriving microstate $j$ at $t^{\prime\prime}$\ under the assumption of the
detailed balance given in Eq. (\ref{DetailedBalance}). As $p_{j\text{eq}%
}(E_{j}^{\prime},\lambda^{\prime})$ is time-independent, $T_{ij}%
(\delta^{\prime\prime})$ is also time indepndent so it must have no dependence
on $\delta^{\prime\prime}$.
\end{claim}

We see that under the assumption of the detailed balance, each row of the
matrix $\mathbf{T}$ is the equilibrium row vector $\mathbf{p}_{\text{eq}%
}(\lambda^{\prime})$:%
\begin{equation}
\overleftrightarrow{\mathbf{T}}=\left(
\begin{array}
[c]{c}%
\mathbf{p}_{\text{eq}}(\mathbf{E}(\lambda^{\prime}),\lambda^{\prime})\\
\mathbf{p}_{\text{eq}}(\mathbf{E}(\lambda^{\prime}),\lambda^{\prime})\\
\vdots\\
\mathbf{p}_{\text{eq}}(\mathbf{E}(\lambda^{\prime}),\lambda^{\prime})
\end{array}
\right)  ,\label{BalancedT-form}%
\end{equation}
here, the equilibrium vector refers to $\mathbf{p}_{\text{eq}}(\mathbf{E}%
(\lambda^{\prime}),\lambda^{\prime})\doteq\left\{  p_{j\text{eq}}%
(E_{j}(\lambda^{\prime}),\lambda^{\prime})\right\}  $ at $t^{\prime\prime}$,
which\ is \emph{invariant} under the transition matrix $\overleftrightarrow
{\mathbf{T}}$:%
\begin{equation}
\mathbf{p}_{\text{eq}}(\mathbf{E}(\lambda^{\prime}),\lambda^{\prime
})=\mathbf{p}_{\text{eq}}(\mathbf{E}(\lambda^{\prime}),\lambda^{\prime
})\overleftrightarrow{\mathbf{T}}\label{BalancedT}%
\end{equation}
for each interval $\delta^{\prime\prime}$. Such a transition matrix is said be
\emph{balanced} \cite{Crooks-PRE-2000} over $\delta^{\prime\prime}$. It is
easy to conclude from the above invariance equation that $\overleftrightarrow
{\mathbf{T}}$ must be a stationary transition matrix. It is known from the
theory of Markov chains that such a matrix is the limiting matrix; see Eq.
(\ref{LimitT}).

While the form of Eqs. (\ref{BalancedT-form}) and (\ref{BalancedT}) is trivial
based on Eq. (\ref{DetailedBalance}), it is well known that a balanced
$\overleftrightarrow{\mathbf{T}}$ in Eq. (\ref{BalancedT}) also has the same
form as in Eq. (\ref{BalancedT-form}); see the Fundamental Limit theorem or
Doeblin's theorem\emph{ }of Markov chains \cite{Stroock,Beichelt}. We simply
state the theorem below:

\begin{theorem}
\label{Theorem-LimitTheorem}\textbf{Fundamental Limit Theorem}: For the
transition matrix $\mathbf{T}$ for a regular Markov chain with finite
number~of states$~(r<\infty)$,
\begin{equation}
\lim_{k\rightarrow\infty}\mathbf{T}^{k}=\overleftrightarrow{\mathbf{T}}
\label{LimitT}%
\end{equation}
given in Eq. (\ref{BalancedT-form}) and satisfying the balanced condition in
Eq. (\ref{BalancedT}). In particular,%
\begin{equation}
\lim_{k\rightarrow\infty}\mathbf{T}_{ij}^{k}=\overleftrightarrow{T}%
_{ij}=p_{j\text{eq}}>0\text{ for }\forall k. \label{LimitT-component}%
\end{equation}

\end{theorem}

\begin{proof}
For proof, see Refs. \cite{Stroock,Beichelt} or any other text book on Markov chains.
\end{proof}

Therefore, from now on, we will take the balanced $\overleftrightarrow
{\mathbf{T}}$ to be given by Eq. (\ref{BalancedT-form}).

We now insert the index $l$ so that\ the matrix elements of such a transition
matrix (recall that $\lambda^{\prime}=\lambda(t_{l+1}^{\prime})=\lambda
(t_{l+1})=\lambda_{l+1}$ and $E_{j}^{\prime}\doteq E_{j}(t_{l+1}^{\prime
})\equiv E_{j}(t_{l+1})=E_{j}(\lambda_{l+1})$) is uniquely determined and
given by%
\begin{equation}
\overleftrightarrow{T}_{i,j}^{(l)}=p_{j\text{eq }}(E_{j}(\lambda
_{l+1}),\lambda_{l+1}),\forall i,l,\label{TransitionMatrix-Unique}%
\end{equation}
where $p_{j\text{eq}}(E_{j},\lambda_{l+1})$ is the equilibrium probability of
the equilibrium $j$th microstate at time $t_{l+1}$; see Eq.
(\ref{BalancedT-form}). Using the above transition matrix $\overleftrightarrow
{\mathbf{T}}^{(l)}$, we can determine how any arbitrary row probability vector
$\mathbf{p}$\ changes over $\delta_{l}^{\prime\prime}$:
\begin{equation}
\mathbf{p}(\mathbf{E}_{l+1},\lambda_{l+1},t_{l+1})=\mathbf{p}(\mathbf{E}%
_{l+1}^{\prime},\lambda_{l+1},t_{l+1}^{\prime})\overleftrightarrow{\mathbf{T}%
}^{(l)},\label{TransitionMatrix-ArbitraryVector-Effect}%
\end{equation}
where $\mathbf{p}(\mathbf{E}_{l+1},\lambda_{l+1},t_{l+1}),\mathbf{E}%
_{l+1}\doteq\{E_{j}(\lambda_{l+1})\}$, is used as a short hand notation for
the vector $\left\{  p_{j}(E_{j}(\lambda_{l+1}),\lambda_{l+1},t_{l+1}%
)\right\}  $ as above. We first recognize that each row in
$\overleftrightarrow{\mathbf{T}}^{(l)}$\ is the same row vector $\mathbf{p}%
_{\text{eq}}(\lambda_{l+1})$ in accordance with Eq. (\ref{BalancedT-form}), a
fact that does not seems to have been recognized in the current literature.
This means that the $j$th column in $\overleftrightarrow{\mathbf{T}}^{(l)}$
has the same entry $p_{j\text{eq}}(\lambda_{l+1})$ for all rows. The effect of
this is the following surprising observation that, after evaluating the right
side of the above equation, we obtain%
\begin{equation}
\mathbf{p}(\mathbf{E}_{l+1},\lambda_{l+1},t_{l+1})=\mathbf{p}_{\text{eq}%
}(\mathbf{E}_{l+1},\lambda_{l+1})\label{TransitionMatrix-Unique-Effect}%
\end{equation}
for any arbitrary row probability vector $\mathbf{p}(\mathbf{E}_{l+1}^{\prime
},\lambda_{l+1},t_{l+1}^{\prime})$, not necessarily an equilibrium vector
$\mathbf{p}_{\text{eq}}(\mathbf{E}_{l+1}^{\prime},\lambda_{l+1})$ at
$t_{l+1}^{\prime}$. Despite the arbitrary $\mathbf{p}(\mathbf{E}_{l+1}%
^{\prime},\lambda_{l+1},t_{l+1}^{\prime})$, the result of $\overleftrightarrow
{\mathbf{T}}^{(l)}$ on it is to yield the equilibrium vector $\mathbf{p}%
_{\text{eq}}(\mathbf{E}_{l+1},\lambda_{l+1})$. We illustrate this by a simple
example for $r=2$ and some arbiter $l$, which we again suppress. The
corresponding $\overleftrightarrow{\mathbf{T}}=\overleftrightarrow{\mathbf{T}%
}^{(l)}$ is written as (we suppress $l$, $E_{j}^{\prime}$ and $\lambda_{l+1}%
$)
\[
\overleftrightarrow{\mathbf{T}}=\left(
\begin{tabular}
[c]{ll}%
$p_{1\text{eq}}$ & $p_{2\text{eq}}$\\
$p_{1\text{eq}}$ & $p_{2\text{eq}}$%
\end{tabular}
\right)  .
\]
We consider an arbitrary probability vector $\mathbf{p}=\left(  p_{1}%
,p_{2}\right)  $ to find that%
\[
\left(  p_{1},p_{2}\right)  \left(
\begin{tabular}
[c]{ll}%
$p_{1\text{eq}}$ & $p_{2\text{eq}}$\\
$p_{1\text{eq}}$ & $p_{2\text{eq}}$%
\end{tabular}
\right)  =\left(  p_{1\text{eq}},p_{2\text{eq}}\right)  ,
\]
where we have used the fact that $p_{1}+p_{2}=1$. This thus justifies using
the the microstate $j$ as an EQ-microstate in the notation
$\overleftrightarrow{T}_{i,j}^{(l)}$ as noted above.

\section{Order for Microwork and Microheat\label{Sec-Work-Heat-Order}}

We have seen in the previous section that the microstate $k_{l}^{\prime}$ at
$t_{l+1}^{\prime}$ is a NEQ-microstate. However, the interaction with
$\Sigma_{\text{h}}$ during $\delta_{l}^{\prime\prime}$ ensures that the
probability at $t_{l+1}$ is the equilibrium probability in accordance with Eq.
(\ref{TransitionMatrix-Unique-Effect}). Thus, we conclude that

\begin{conclusion}
\label{Conclusion-Microstate-Nature} The microstates $k_{l}(t_{l})$ and
$k_{l+1}(t_{l+1})$ are EQ-microstates, but the intermediate microstate
$k_{l}^{\prime}(t_{l+1}^{\prime})$ is not.
\end{conclusion}

The order of the interactions with $\Sigma_{\text{w}}$ and $\Sigma_{\text{h}}$
during $\delta_{l}$ is \ important to ensure that $k_{l}(t_{l})$ and
$k_{l+1}(t_{l+1})$ remain EQ-microstates. Suppose we interchange their order
so that the interaction with $\Sigma_{\text{h}}$ occurs during $\delta
_{l}^{\prime}$, followed by the interaction with $\Sigma_{\text{w}}$ during
$\delta_{l}^{\prime\prime}$ with the final effect $\lambda_{l}(t_{l}%
)\rightarrow\lambda_{l+1}(t_{l+1}),k_{l}^{\prime}(t_{l})\rightarrow
k_{l+1}^{\prime}(t_{l+1})$ at the end of $\delta_{l}$; here, we have used the
fact to be established below that $k_{l}^{\prime}(t_{l})$ and $k_{l+1}%
^{\prime}(t_{l+1})$ are NEQ-microstates so we will denote their energies with
a prime in the following. As the change $\lambda_{l}\rightarrow\lambda_{l+1}$
is a consequence of the $\Sigma_{\text{w}}$-interaction, the $\Sigma
_{\text{h}}$-interaction resulting in $k_{l}^{\prime}(t_{l})\rightarrow
k_{l+1}(t_{l+1}^{\prime})$ (as shown below, $k_{l+1}(t_{l+1}^{\prime})$ is an
EQ-microstate after this interaction so its energy is denoted without a prime)
must occur at fixed $\lambda_{l}$. From the discussion above, we conclude that
for the transition $i^{\prime}=k_{l}^{\prime}(t_{l})\rightarrow j=k_{l+1}%
(t_{l+1}^{\prime})$, the matrix element $\overleftrightarrow{T}_{ij}%
^{(l)}(\delta_{l}^{\prime})$ (we use $i$ instead of $i^{\prime}$ in the
subscript) is given by
\begin{equation}
\overleftrightarrow{T}_{ij}^{(l)}(\delta_{l}^{\prime})=p_{j\text{eq }}%
(E_{j}(\lambda_{l}),\lambda_{l}),\forall i^{\prime}%
,l,\label{TransitionMatrix-Interchange}%
\end{equation}
which is different from the matrix elements in Eq.
(\ref{TransitionMatrix-Unique}); here $E_{j}(\lambda_{l})$ is the energy of
$j(t_{l+1}^{\prime})$, which is the same as the energy $E_{j}^{\prime}%
(\lambda_{l})\ $at $t_{l}$. As a consequence, Eq.
(\ref{TransitionMatrix-Unique-Effect}) is replaced by%
\begin{equation}
\mathbf{p}(\mathbf{E}(\lambda_{l}),\lambda_{l},t_{l+1}^{\prime})=\mathbf{p}%
(\mathbf{E}^{\prime}(\lambda_{l}),\lambda_{l},t_{l})\overleftrightarrow
{\mathbf{T}}^{(l)}=\mathbf{p}_{\text{eq}}(\mathbf{E}(\lambda_{l}),\lambda
_{l}).\label{TransitionMatrix-Unique-Interchange}%
\end{equation}
This means that the probability of the arriving microstate $j$ with
$\lambda(t_{l+1}^{\prime})=\lambda_{l}$\ is the equilibrium probability of $j$
at $t_{l+1}^{\prime}$, even though $i^{\prime}$ at the start of $\delta_{l}$
is a NEQ-microstates as we now demonstrate. 

The $\Sigma_{\text{w}}$-interaction with $k_{l+1}$ at $t_{l+1}^{\prime}$
results in $\lambda_{l}\rightarrow\lambda_{l+1}$ and changes its energy to
$E_{k_{l+1}}^{\prime}(\lambda_{l+1})$ from $E_{k_{l+1}}(\lambda_{l})$ due to
the microwork $\Delta R_{\text{C}}^{(l)}(k_{l+1},\lambda_{l}\rightarrow
\lambda_{l+1})$ without changing its probability so that $p_{k_{l+1}%
}(E_{k_{l+1}}^{\prime}(\lambda_{l+1}),\lambda_{l+1},t_{l+1})=p_{k_{l+1}%
\text{eq}}(E_{k_{l}}(\lambda_{l}),\lambda_{l})\neq p_{k_{l+1}\text{eq}%
}(E_{k_{l+1}}^{\prime}(\lambda_{l+1}),\lambda_{l+1})$. This explains why the
resulting NEQ-microstate at $t_{l+1}$ is denoted by $k_{l+1}^{\prime}$. Using
the same argument by replacing $l$ by $l-1$,\ we will conclude that the final
microstate $k_{l}^{\prime}$ at the end of $\delta_{l-1}$\ is also a
NEQ-microstate, which finally proves our assertion that both $k_{l}^{\prime}$
and $k_{l+1}^{\prime}$ are NEQ-microstates, with the intermediate microstate
$k_{l+1}$ an EQ-microstate.

It should be evident that in order to obtain an EQ terminal microstate $k_{n}%
$, we must always have the $\Sigma_{\text{w}}$-interaction followed by the
$\Sigma_{\text{h}}$-interaction during $\delta_{n-1}$. Similarly, if we wish
to start from an EQ-microstate $k_{0}$, we must use the same order of the two
interactions in $\delta_{0}$. Recursively, this is true for all other
intervals, which means that the order of the two interactions in each interval
$\delta_{l}$ \emph{must not be reversed}. Thus, we come to the following

\begin{conclusion}
\label{Conclusion-InteractionOrder}In each interval $\delta_{l}$, we must have
the $\Sigma_{\text{w}}$-interaction followed by the $\Sigma_{\text{h}}$-interaction.
\end{conclusion}

\section{Consequences of Detailed Balance\label{Sec-Consequences}}

We now follow the consequence of Eq. (\ref{TransitionMatrix-Unique-Effect}).
We start with $l=0$ and focus on the initial EQ-microstate $k_{0}$ during the
interval $\delta_{0}$\ of $\overline{\gamma}^{\text{(F)}}$; see Eq.
(\ref{Forward-trajectory}). It is turned into a NEQ-microstate $k_{0}^{\prime
}$ at $t_{1}^{\prime}$ after microwork has been performed during $\delta
_{0}^{\prime}$ so that $E_{k_{0}}^{\prime}(\lambda_{1})\doteq E(k_{0}%
,\lambda_{0})+\Delta R_{\text{C}}^{(0)}(k_{0},\lambda_{0}\rightarrow
\lambda_{1})$ at $t_{1}^{\prime}$, which we will simply denote by $E_{k_{0}%
}^{\prime}(\lambda_{1})$ so that we can use the notation $\mathbf{E}^{\prime
}(\lambda_{1})$ for the set $\{E_{k_{0}^{\prime}}^{\prime}(\lambda_{1})\}$;
recall that we suppress the prime on the subscript $k_{0}^{\prime}$ in
$E_{k_{0}}^{\prime}(\lambda_{1})$. The probability of the initial microstate
$k_{0}$ is $p_{k_{0}\text{eq}}(E_{k_{0}},\lambda_{0})$. At $t_{1}^{\prime}$,
$k_{0}^{\prime}$ has the same probability as this initial probability:
$p_{k_{0}^{\prime}}(E_{k_{0}^{\prime}},\lambda_{1})=p_{k_{0}\text{eq}%
}(E_{k_{0}},\lambda_{0})$. Interaction with $\Sigma_{\text{h}}$\ does not
change the energies of the microstates but only their probabilities.
Therefore, $E(k_{0},\lambda_{1})=E(k_{0}^{\prime},\lambda_{1})\neq
E(k_{0},\lambda_{0})$ but $E(\overline{k}_{0},\lambda_{1})\equiv
E(\overline{k}_{0},\lambda_{0}),\forall\overline{k}_{0}\neq$ $k_{0}$ at time
$t_{1}$; in particular, $E(k_{1},\lambda_{1})\equiv E(k_{1},\lambda_{0})$.
This is because  as no microwork is done on these microstates $\neq k_{0}$.
Thus, $\mathbf{E}^{\prime}(\lambda_{1})=\mathbf{E}(\lambda_{1})=\mathbf{E}%
(\lambda_{0})$, except for the initial microstate $k_{0}.$ As we have seen in
the previous section, $\overleftrightarrow{\mathbf{T}}^{(0)}$ is determined by
$\mathbf{p}_{\text{eq}}(\mathbf{E}^{\prime},\lambda_{1})$ and for which
$\mathbf{p}_{\text{eq}}(\mathbf{E},\lambda_{0})$ corresponding to the initial
macrostate is not an invariant probability vector. Therefore, $\mathbf{p}%
_{\text{eq}}(\mathbf{E},\lambda_{0})$ represents an arbitrary probability
vector in Eq. (\ref{TransitionMatrix-ArbitraryVector-Effect}). It now follows
from Eq. (\ref{TransitionMatrix-Unique-Effect}) that%
\[
\mathbf{p}(\mathbf{E}^{\prime},\lambda_{1},t_{1}^{\prime})\overleftrightarrow
{\mathbf{T}}^{(0)}=\mathbf{p}_{\text{eq}}(\mathbf{E},\lambda_{1}).
\]
We also know that $\mathbf{p}(\mathbf{E}^{\prime},\lambda_{1},t_{1}^{\prime
})\equiv\mathbf{p}_{\text{eq}}(\mathbf{E},\lambda_{0})$, the initial
probability vector of the EQ-macrostate \textsf{A. }Therefore, we can
equivalently say that the effect of $\overleftrightarrow{\mathbf{T}}^{(0)}$ on
$\mathbf{p}_{\text{eq}}(\mathbf{E},\lambda_{0})$ results in $\mathbf{p}%
_{\text{eq}}(\mathbf{E}^{\prime},\lambda_{1})$, which is the equilibrium
probability vector for the arriving macrostate at the end $t_{1}=\delta t$ of
$\delta_{0}$ or the start of $\delta_{1}$. This is the macrostate that is
arrived at after the successive interactions, first with $\Sigma_{\text{w}}$
and then with $\Sigma_{\text{h}}$. It is obviously an EQ-macrostate as its
probability vector is $\mathbf{p}_{\text{eq}}(\mathbf{E},\lambda_{1})$. It
arises when we consider all the trajectories starting with any of the $r$
possible microstates $\{$\textsf{a}$_{k}\}$ for $k_{0}$. In other words, all
possible microstates at $t_{1}$ represent EQ-microstates, \textit{i.e.}, an
EQ-macrostate. However, we should not forget that $\overleftrightarrow
{\mathbf{T}}^{(0)}$ only acts during $\delta_{0}^{\prime}$ and not over the
entire duration $\delta_{0}$. This also means that $k_{0}$ does not have to an
EQ-microstate or that \textsf{A} does not have to be an EQ-macrostate for the
macrostate at $t=t_{1}$ to be an EQ-macrostate.

It now follows that applying Eq. (\ref{TransitionMatrix-Unique-Effect}) along
with the above argument for $l=1,2,\cdots,n-1$ successively, we see that
\begin{equation}
\mathbf{p}(\mathbf{E}_{l+1},\lambda_{l+1},t_{l+1})=\mathbf{p}_{\text{eq}%
}(\mathbf{E}_{l+1},\lambda_{l+1}),l=0,1,2,\cdots,n-1,\label{Probability-Heat}%
\end{equation}
where $\mathbf{E}_{l+1}$ really stands for \ $\mathbf{E}(\lambda_{l+1}%
)\equiv\mathbf{E}(\lambda_{l+1})$ at the end of the interval $\delta_{l}$. We
are now set for drawing one of our most important conclusions based on the
principle of detailed balance or the use of the particular transition matrix
$\overleftrightarrow{\mathbf{T}}^{(l)}$:

\begin{conclusion}
\label{Conclusion-EqStates}Due to the requirement of the principle of detailed
balance and having the $\Sigma_{\text{h}}$-interaction always occur in the
second-half of each interval $\delta_{l}$, the successive probability vectors
at times $t_{m},m=1,\cdots,n$ are $\mathbf{p}_{\text{eq}}(\mathbf{E}%
(\lambda_{1}),\lambda_{1}),\mathbf{p}_{\text{eq}}(\mathbf{E}(\lambda
_{2}),\lambda_{2}),\cdots,\mathbf{p}_{\text{eq}}(\mathbf{E}(\lambda
_{n}),\lambda_{n})$, respectively, for the process $\mathcal{P}^{\text{(F)}}$;
we do not have to assume that \textsf{A }must be an EQ-macrostate. Thus, all
intermediate macrostates and the final macrostate \textsf{B} must be
EQ-macrostates, contrary to what is usually claimed.
\end{conclusion}

By requiring the initial macrostate \textsf{A }to be also an EQ-macrostate, we
can also add $\mathbf{p}_{\text{eq}}(\mathbf{E}(\lambda_{0}),\lambda_{0})$ in
the above sequence, but this requirement is not essential. However, the most
important conclusion is about the final macrostate \textsf{B}\ being an
EQ-macrostate. As EQ-macrostates imply EQ-microstates, the above Conclusion
means that \emph{not only the last microstate but every intermediate
microstate in the Crooks' approach is an equilibrium microstate}, which
contradicts the common understanding of the Crooks approach in which the final
macrostate does not have to be an EQ-macrostate as is the case with the
Jarzynski process. This also justifies why $\tau$ should really be thought of
as representing $\tau_{\text{eq}}$.

\section{Backward Trajectory and determination of $\omega(\overline{\gamma
}_{\text{\textsf{ab}}}^{\text{(F)}})$\label{Sec-BackwardTrajectory}}

\subsection{Limitations of the Crooks' Approach}

Crooks identifies the forward trajectory by the sequence of microstates
\begin{subequations}
\begin{equation}
\overline{\gamma}_{\text{C}}^{\text{(F)}}:k_{0}\overset{\lambda_{1}%
}{\rightarrow}k_{1}\overset{\lambda_{2}}{\rightarrow}k_{2}\overset{\lambda
_{3}}{\rightarrow}k_{3}\cdots k_{n-1}\overset{\lambda_{n}}{\rightarrow}k_{n},
\label{CrooksForwardTrajectory}%
\end{equation}
and defines the backward trajectory as%
\begin{equation}
\overline{\gamma}_{\text{C}}^{\text{(B)}}:k_{0}\overset{\lambda_{1}%
}{\leftarrow}k_{1}\overset{\lambda_{2}}{\leftarrow}k_{2}\overset{\lambda_{3}%
}{\leftarrow}k_{3}\cdots k_{n-1}\overset{\lambda_{n}}{\leftarrow}k_{n},
\label{CrooksBackwardTrajectory}%
\end{equation}
which seems a natural choice by reversing the transitions; we have suppressed
the suffix \textsf{ab} for simplicity. The backward trajectory for a Markov
chain has been an actively investigated topic \cite{Kelly}. We will simply
quote some relevant results. The transition matrix element $T_{ij}%
^{\text{(B)}}$ for the backward trajectory is given by%
\end{subequations}
\begin{equation}
\overleftrightarrow{T}_{ij}^{\text{(B)}}=\frac{\overleftrightarrow{T}%
_{ji}^{\text{(F)}}p_{j\text{eq}}}{p_{i\text{eq}}}=p_{j\text{eq}}%
=\overleftrightarrow{T}_{ij}^{\text{(F)}}, \label{BackwardTransitionElement}%
\end{equation}
where we have used twice Eq. (\ref{TransitionMatrix-Unique}), which is valid
in the Crooks's approach. Such a Markov chain for which the forward and
backward trajectories (chains) have the same transition probabilities is said
to be \emph{reversible}, and there is no need to use the two superscripts F
and B on $\overleftrightarrow{T}_{ij}$. In this case, the principle of
detailed balance is satisfied. Thus, the Crooks approach results in a
reversible trajectory. According to \emph{Kolmogorov's criterion }\cite[p.
21]{Kelly}, we must have%
\begin{align}
&  \overleftrightarrow{T}_{k_{0}k_{1}}\overleftrightarrow{T}_{k_{1}k_{2}%
}\cdots\overleftrightarrow{T}_{k_{n-1}k_{n}}\overleftrightarrow{T}_{k_{n}%
k_{0}}\nonumber\\
&  =\overleftrightarrow{T}_{k_{0}k_{n}}\overleftrightarrow{T}_{k_{n}k_{n-1}%
}\cdots\overleftrightarrow{T}_{k_{2}k_{1}}\overleftrightarrow{T}_{k_{1}k_{0}}.
\label{KolmogorovCriteria}%
\end{align}
This merely expresses the fact that by adding the initial microstate $k_{0}$
after $k_{n}$ and pictorially thinking of all the microstates to be on a
closed loop, we can traverse the loop from $k_{0}$ in either directions, one
which we identify as the forward and the other backward. Therefore, the
equivalence of the two transition matrices allows us to think of the backward
trajectory as the \emph{continuation }of the forward trajectory with
additional microstates $\left\{  \mathsf{a}_{k}\right\}  $ with the microstate
index $k=n+1,n+2,\cdots,2n$ and use the transformation $k^{\prime}=2n-k$ to
put them on a loop as shown below:%
\[
k_{0}%
\begin{array}
[c]{c}%
_{\overset{\lambda_{1}}{\nearrow}}k_{1}\overset{\lambda_{2}}{\rightarrow}%
k_{2}\overset{\lambda_{3}}{\rightarrow}k_{3}\cdots k_{n-1}{}_{\overset
{\lambda_{n}}{\searrow}}\\
^{\overset{\lambda_{1}}{\nwarrow}}k_{1}\overset{\lambda_{2}}{\leftarrow}%
k_{2}\overset{\lambda_{3}}{\leftarrow}k_{3}\cdots\leftarrow k_{n-1}%
^{\overset{\lambda_{n}}{\swarrow}}%
\end{array}
k_{n},
\]
except for the mismatch for the index for the work parameter $\lambda$. The
transformation $k\rightarrow k^{\prime}$ would have required $\lambda_{n+1}$
for $k_{n}\overset{\lambda_{n+1}}{\rightarrow}k_{n+1}$ to be transformed to
$\lambda_{n-1}$ and to the transition $k_{n-1}\overset{\lambda_{n-1}%
}{\longleftarrow}k_{n}$. But Crooks takes it to be $k_{n-1}\overset
{\lambda_{n}}{\longleftarrow}k_{n}$, see Eq. (\ref{CrooksBackwardTrajectory}).
This casts doubt on the backward trajectory in Eq.
(\ref{CrooksBackwardTrajectory}) not forming a reversible trajectory, which
would then violate Kolmogorov's criterion.

The significance of the transition$\ k_{l}\overset{\lambda_{l+1}}{\rightarrow
}k_{l+1}$ or its reverse$\ k_{l}\overset{\lambda_{l+1}}{\leftarrow}k_{l+1}$ is
that the work parameter $\lambda_{l+1}$ is kept fixed during the transition
$k_{l}\rightarrow k_{l+1}$ or $k_{l}\leftarrow k_{l+1}$. It is with this
choice that the CFT in Eq. (\ref{Crooks-FT0}) has been derived. It should be
evident from the discussion earlier that each transition refers to a
$\Sigma_{\text{h}}$-interaction, whether the transition is forward or
backward. Since we have required the initial microstate $k_{0}~$of
$\overline{\gamma}^{\text{(F)}}$\ to be an EQ-microstate, we need to follow
Conclusion \ref{Conclusion-InteractionOrder} so that each $\Sigma_{\text{h}}%
$-interaction is preceded by a $\Sigma_{\text{w}}$-interaction in each
interval $\delta_{l}$, even for the continuation of the forward trajectory
described above. This is consistent with the first transition $k_{0}%
\overset{\lambda_{1}}{\rightarrow}k_{1}$ in \ Eq.
(\ref{CrooksForwardTrajectory}), where we fix $\lambda=\lambda_{1}$ after
$\lambda_{0}\rightarrow\lambda_{1}$ due to the $\Sigma_{\text{w}}%
$-interaction. This is true of all the transitions in Eq.
(\ref{CrooksForwardTrajectory}). Therefore, according to Conclusion
\ref{Conclusion-EqStates}, the end microstate $k_{n}$ with $\lambda
=\lambda_{n}$ for $\overline{\gamma}^{\text{(F)}}$ is an EQ-microstate.

We are now ready to introduce the backward trajectory $\overline{\gamma
}^{\text{(B)}}$ for which the EQ-microstate $k_{n}$ with $\lambda=\lambda_{n}$
must play the role of the initial microstate. The final microstate of the
backward trajectory must be the initial EQ-microstate $k_{0}$ of
$\overline{\gamma}^{\text{(F)}}$. This allows us to reverse the sequence
$(k_{0},k_{1},k_{2},\cdots,k_{n-1},k_{n})$ for $\overline{\gamma}^{\text{(F)}%
}$ to the sequence $(k_{n},k_{n-1},k_{n-2},\cdots,k_{1},k_{0})$ for
$\overline{\gamma}^{\text{(B)}}$ as used in the definition of $e^{\omega
(\overline{\gamma}_{\text{C,\textsf{ab}}}^{\text{(F)}})}\ $in Eq.
(\ref{Crooks-FT0}). Both trajectories sample the same two EQ-macrostates
\textsf{A }and \textsf{B} in reverse order by going through the same
intermediate microstates.

Let us first consider the backward trajectory $\overline{\gamma}_{\text{C}%
}^{\text{(B)}}$ specified in Eq. (\ref{CrooksBackwardTrajectory}), which does
not mention the $\Sigma_{\text{w}}$-interactions, but which can be deduced
from the above specification of $\overline{\gamma}_{\text{C}}^{\text{(B)}}$.
It is clearly seen from $k_{n-1}\overset{\lambda_{n}}{\leftarrow}k_{n}$ at
fixed $\lambda=\lambda_{n}$ that it represents a $\Sigma_{\text{h}}%
$-interaction, which then must be followed by a $\Sigma_{\text{w}}%
$-interaction in which $\lambda_{n}\rightarrow\lambda_{n-1}$ so that we can
follow the next $\Sigma_{\text{h}}$-interaction specified by $k_{n-2}%
\overset{\lambda_{n-1}}{\leftarrow}k_{n-1}$. We see that the order of the two
interactions are, therefore, reversed for $\overline{\gamma}_{\text{C}%
}^{\text{(B)}}$ compared to that for $\overline{\gamma}_{\text{C}}%
^{\text{(F)}}$. Following this argument recursively, we come to the last
$\Sigma_{\text{h}}$-interaction specified by $k_{0}\overset{\lambda_{1}%
}{\leftarrow}k_{1}$ for $\overline{\gamma}_{\text{C}}^{\text{(B)}}$. This
cannot be the last interaction as it will leave the value of the work
parameter at $\lambda=\lambda_{1}$, while $\lambda=\lambda_{0}$ for the
macrostate \textsf{A}. Therefore, there must be another $\Sigma_{\text{w}}%
$-interaction to ensure that $\lambda_{1}\rightarrow\lambda_{0}$. However,
this makes the last microstate a NEQ-microstate $k_{0}^{\prime}$ whose
probability is not the equilibrium probability $p_{k_{0}\text{eq}}(E_{k_{0}%
},\lambda_{0})$. Thus, the proposed backward trajectory $\overline{\gamma
}_{\text{C}}^{\text{(B)}}$ does not bring back the system to the EQ-macrostate
\textsf{A}.

\subsection{New Approach}

\begin{claim}
\label{Conclusion-A-EQ-Macrostate}In order to define a backward trajectory
properly, it is important to have \textsf{A} an EQ-macrostate.
\end{claim}

Therefore, we will assume as is commonly done that \textsf{A }is an EQ-macrostate.

We have uniquely identified the forward trajectory in Eq.
(\ref{Forward-trajectory}) involving the intermediate NEQ-microstates
$\{k_{k}^{\prime}\}$, and can symbolically represent as an interaction
sequence (we suppress w and h above the arrows as they can be uniquely
deciphered from the sequence)
\begin{subequations}
\begin{equation}
\overline{\gamma}^{\text{(F)}}:k_{0}\rightarrow k_{0}^{\prime}\rightarrow
k_{1}\rightarrow k_{1}^{\prime}\cdots k_{n-1}\rightarrow k_{n-1}^{\prime
}\rightarrow k_{n}.\label{ForwardTrajectorySequence}%
\end{equation}
This trajectory starts at the EQ-microstate $k_{0}$\ and terminates in the
EQ-microstate $k_{n}$. We can now uniquely identify the backward trajectory
$\overline{\gamma}^{\text{(B)}}$ that starts in the EQ-microstate $k_{n}$ and
terminates in the EQ-microstate $k_{0}$ by the pictorial continuation
described above as follows:%
\begin{equation}
\overline{\gamma}^{\text{(B)}}:k_{0}\leftarrow k_{1}^{\prime}\leftarrow
k_{1}\leftarrow k_{2}^{\prime}\leftarrow\cdots k_{n-1}\leftarrow k_{n}%
^{\prime}\leftarrow k_{n}.\label{ReverseTrajectorySequence}%
\end{equation}
By suppressing all of the $\Sigma_{\text{w}}$-interactions, we can rewrite the
above sequences as%
\end{subequations}
\begin{equation}%
\begin{array}
[c]{c}%
\overline{\gamma}^{\text{(F)}}:k_{0}^{\prime}\overset{\lambda_{1}}%
{\rightarrow}k_{1},k_{1}^{\prime}\overset{\lambda_{2}}{\rightarrow}%
k_{2},\cdots,k_{n-2}^{\prime}\overset{\lambda_{n-1}}{\rightarrow}%
k_{n-1},k_{n-1}^{\prime}\overset{\lambda_{n}}{\rightarrow}k_{n},\\
\overline{\gamma}^{\text{(B)}}:k_{0}\overset{\lambda_{0}}{\leftarrow}%
k_{1}^{\prime},k_{1}\overset{\lambda_{1}}{\leftarrow}k_{2}^{\prime}%
,\cdots\overset{\lambda_{n-2}}{,k_{n-2}\leftarrow}k_{n-1}^{\prime}%
,k_{n-1}\overset{\lambda_{n-1}}{\leftarrow}k_{n}^{\prime}.
\end{array}
\label{TrajectorySequence-Condensed}%
\end{equation}
If we compare the above sequences with the sequences in Eqs.
(\ref{CrooksForwardTrajectory}) and (\ref{CrooksBackwardTrajectory}), we
observe that while the sequence of the work parameters are the same for the
forward trajectories, they are displaced by one unit for the backward
trajectories. However, $\overline{\gamma}^{\text{(B)}}$ brings the system back
to the EQ-microstate $k_{0}$ with $\lambda=\lambda_{0}$\ as required.

\begin{remark}
We see that the order of the two interactions in Conclusion
\ref{Conclusion-InteractionOrder} must also be maintained in the backward
trajectories to ensure that the backward trajectories bring the system back to
the initial EQ-microstate $k_{0}$.
\end{remark}

This also means that the backward trajectories bring the system back to the
EQ-macrostate \textsf{A}.

We are now prepared to evaluate $\omega(\overline{\gamma}^{\text{(F)}})$ that
is defined similar to the definition of $\omega_{\text{C}}(\overline{\gamma
}_{\text{C}}^{\text{(F)}})$ in Eq. (\ref{Crooks-FT0}):%
\[
e^{\omega(\overline{\gamma}^{\text{(F)}})}\doteq\frac{p(\overline{\gamma
}_{\text{\textsf{ab}}}^{\text{(F)}})}{p(\overline{\gamma}_{\text{\textsf{ba}}%
}^{\text{(B)}})}%
\]
The probability $p(\overline{\gamma}^{\text{(F)}})$ is%
\[
p_{k_{0\text{eq}}}(E_{0},\lambda_{0})p_{k_{1\text{eq}}}(E_{1},\lambda
_{1})p_{k_{2\text{eq}}}(E_{2},\lambda_{2})\cdots p_{k_{n\text{eq}}}%
(E_{n},\lambda_{n}),
\]
where we have used the transition matrix elements given in Eq.
(\ref{TransitionMatrix-Unique}). Comparing this with the left side of Eq.
(\ref{KolmogorovCriteria}) and recalling that $\overleftrightarrow{T}%
_{k_{n}k_{0}}=p_{k_{0\text{eq}}}(E_{0},\lambda_{0})$, we see that they are
identical. We similarly find for $p(\overline{\gamma}^{\text{(B)}})\ $the
value
\[
p_{k_{n\text{eq}}}(E_{n},\lambda_{n})p_{k_{n-1\text{eq}}}(E_{n-1}%
,\lambda_{n-1})\cdots p_{k_{0\text{eq}}}(E_{0},\lambda_{0}),
\]
which is the same as for $p(\overline{\gamma}^{\text{(F)}})$ and the right
side of Eq. (\ref{KolmogorovCriteria}). Therefore, we finally obtain a new FT
\begin{equation}
e^{\omega(\overline{\gamma}^{\text{(F)}})}\equiv1,\forall\overline{\gamma
}^{\text{(F)}}\in\boldsymbol{\overline{\gamma}}_{\text{\textsf{ab}}%
}^{\text{(F)}}, \label{FT-Corrected}%
\end{equation}
which is different from the CFT derived by Crooks. Our derivation also shows
that%
\begin{equation}
\beta_{0}\Delta Q_{\text{C}}(\overline{\gamma}^{\text{(F)}})\equiv\Delta
S(\overline{\gamma}^{\text{(F)}}),\forall\overline{\gamma}^{\text{(F)}}.
\label{CrooksHeatFunction}%
\end{equation}
As $\Delta S(\overline{\gamma}^{\text{(F)}})$, see Eq.
(\ref{Trajectory entropy diff}), is different for different choices of
$k_{0},k_{n}\in(1,2,\cdots,r)$ but the same for all forward trajectories
$\overline{\gamma}^{\text{(F)}}\in$$\boldsymbol{\overline{\gamma}}%
$$_{\text{\textsf{ab}}}^{\text{(F)}}$ between $k_{0}$ and $k_{n}$, it is not a
function of the trajectories but only of $k_{0}$ and $k_{n}$. Thus, it follows
from Eq. (\ref{CrooksHeatFunction}) that $\Delta Q_{\text{C}}(\overline
{\gamma}^{\text{(F)}})$ is also a function only of $k_{0}$ and $k_{n}$ but not
of various $\overline{\gamma}^{\text{(F)}}\in$$\boldsymbol{\overline{\gamma}}%
$$_{\text{\textsf{ab}}}^{\text{(F)}}$. As $\omega(\overline{\gamma
}^{\text{(F)}})\equiv0$ is supposed to denote the microentropy $\Delta
S_{0}(\overline{\gamma}^{\text{(F)}})\equiv\Delta_{\text{i}}S(\overline
{\gamma}^{\text{(F)}})$ of $\Sigma_{0}$\ for \emph{all} possible forward
trajectories, we have
\[
\omega(\overline{\gamma}^{\text{(F)}})\equiv\Delta_{\text{i}}S(\overline
{\gamma}^{\text{(F)}})\equiv0,\forall\overline{\gamma}^{\text{(F)}}%
\in\boldsymbol{\overline{\gamma}}_{\text{\textsf{ab}}}^{\text{(F)}},
\]
where $\Delta_{\text{i}}S(\overline{\gamma}^{\text{(F)}})$ is the internally
generated microentropy in the system over $\overline{\gamma}^{\text{(F)}}$. We
thus conclude that

\begin{enumerate}
\item the thermodynamic average $\Delta_{\text{i}}S^{\text{(F)}}$ of
$\Delta_{\text{i}}S(\overline{\gamma}^{\text{(F)}})$\ also vanishes for the
entire process $\mathcal{P}^{\text{(F)}}$;

\item there are no fluctuations in $\Delta_{\text{i}}S(\overline{\gamma
}^{\text{(F)}})$ over all possible forward trajectories as $\Delta_{\text{i}%
}S(\overline{\gamma}^{\text{(F)}})\equiv0,\forall\overline{\gamma}%
^{\text{(F)}}$;

\item it is easy to establish that the same two conclusions are also valid for
all backward trajectories in $\boldsymbol{\overline{\gamma}}$%
$_{\text{\textsf{BA}}}^{\text{(B)}}$.
\end{enumerate}

It is remarkable that $\Delta_{\text{i}}S(\overline{\gamma}^{\text{(F,B)}})$
is neither positive nor negative but identically zero. The most certain
consequence of $\Delta_{\text{i}}S(\overline{\gamma}^{\text{(F,B)}})\equiv0$
is that there is no irreversible entropy generated during $\mathcal{P}$:
$\Delta_{\text{i}}S^{\text{(F,B)}}\equiv0$. Hence, the CFT only covers
reversible processes; it cannot cover irreversible processes. However,
$\Delta_{\text{i}}S^{\text{(F,B)}}\equiv0$ is also consistent with some
$\Delta_{\text{i}}S(\overline{\gamma}^{\text{(F,B)}})$ being positive and
negative in such a way that the averages vanish. This does not happen in the
new FT in Eq. (\ref{FT-Corrected}) due to the absence of fluctuations in
$\Delta_{\text{i}}S(\overline{\gamma}^{\text{(F,B)}})$.

\section{Discussion and Conclusions\label{Sec-Conclusions}}

The derivation of the CFT is based on two assumptions, the first of which has
never been mentioned to date:

\begin{enumerate}
\item[\textbf{A}1.] the energy change $\Delta E($\textsf{a}$,\lambda)\doteq
E($\textsf{a}$,\lambda+\Delta\lambda)-E($\textsf{a}$,\lambda)$ of a microstate
\textsf{a} due to the change $\Delta\lambda$ in the work parameter $\lambda$
is equal to the exchange work $\Delta R_{\text{C}}($\textsf{a}$,\lambda
\rightarrow\lambda+\Delta\lambda)$, see Remark \ref{Remark-DR-DH-relation} and
Eq. (\ref{Crooks-Work}), and not the generalized work $[-\Delta W($%
\textsf{a}$,\lambda\rightarrow\lambda+\Delta\lambda)]$, see Eq.
(\ref{Gujrati-Work});

\item[\textbf{A}2.] the transition between microstates $k$ and $j$ forms a
Markov chain described by a balanced transition matrix, which has the form
given in Eq. (\ref{BalancedT-form}).
\end{enumerate}

The two assumptions seem to be independent. But this is not the case under
closer scrutiny as we show now. It follows from Assumption \textbf{A}1 that
Crooks does not allow for the possibility of any force imbalance, which is
equivalent to assuming that the internal microwork $\Delta_{\text{i}%
}W_{\mathsf{a}}\doteq\Delta_{\text{i}}W($\textsf{a}$,\lambda\rightarrow
\lambda+\Delta\lambda)=0$. It follows from Eq.
(\ref{micro-internal-heat-workEquality}), which is the microscopic version of
Eq. (\ref{diQ-diW-equality}) that $\Delta_{\text{i}}Q_{\mathsf{a}}=0$ so that
$\Delta Q_{\mathsf{a}}=\Delta_{\text{e}}Q_{\mathsf{a}}$. In other words, there
is no irreversible heat exchange. This is true of the two interactions within
each interval $\delta_{l}$. To ensure $\Delta_{\text{i}}Q_{\mathsf{a}}=0$, we
need to impose the requirement of a balanced transition matrix in each
$\delta_{l}$. If, on the other hand, we allow force imbalance so that
$\Delta_{\text{i}}W_{\mathsf{a}}\neq0$ by relaxing \textbf{A}1, then we cannot
continue to use balanced transition matrix, since that would ensure
$\Delta_{\text{i}}Q_{\mathsf{a}}=0$, which would then violate the
thermodynamic requirement $\Delta_{\text{i}}W_{\mathsf{a}}=\Delta_{\text{i}%
}Q_{\mathsf{a}}$; see Eq. (\ref{micro-internal-heat-workEquality}). Thus,
\textbf{A}2 must also be relaxed. This thus proves our claim.

Note that the above conclusion about the absence of irreversibility refers to
the forward trajectories and, by extension, to backward trajectories
independently and has nothing to do with the Kolmogorov criterion in Eq.
(\ref{KolmogorovCriteria}) or the derivation of the CFT, although they are
both consistent with the lack of irreversibility. In other words, the approach
taken by Crooks cannot capture any irreversibility so it is limited to
reversible processes only.

In summary, we have shown that the use of the principle of detailed balance
results in the microstates $k_{1},k_{2},\cdots,k_{n}$ being EQ-microstates,
even if the initial microstate $k_{0}$ is not an EQ-microstate. However, per
convention, we ensure that even $k_{0}$ is an EQ-microstate. This consequence
of the principle of detailed balance by itself does not mean that the process
during each interval $\delta_{l}$\ is reversible. We can use the above
discussion regarding \textbf{A}1 and \textbf{A}2 to show
\[
\Delta_{\text{i}}W_{\mathsf{a}}\equiv\Delta_{\text{i}}Q_{\mathsf{a}}\equiv0.
\]
Instead, we have followed Crooks and look at the backward trajectories. We
find that the definition of the backward trajectory $\overline{\gamma
}_{\text{C}}^{\text{(B)}}$ fails to reproduce the EQ-microstate $k_{0}$. We
provide a modified version of the backward trajectory $\overline{\gamma
}^{\text{(B)}}$ that ensures to reproduce the EQ-microstate $k_{0}$. We then
find that the ratio of the probabilities of the forward and backward
trajectories turns out to be unity, in accordance with the Kolmogorov
criterion for a reversible Markov chain. From this, we finally conclude that
the adimensional Crooks microheat function $\beta_{0}\Delta Q_{\text{C}%
}(\overline{\gamma}^{\text{(F)}})\equiv\Delta_{0}S(\overline{\gamma
}^{\text{(F)}})$ so that the average irreversible entropy generation
$\Delta_{\text{i}}S^{\text{(F)}}$ vanishes precisely so that the CFT\ is a
result valid only for a reversible process and not for an irreversible
process. This result is consistent with our previous result that the JE is
also valid only for a reversible process.

Let us try to understand the reason for the above limitation of the modified
CFT. There are two steps that are responsible for the new result.

\begin{enumerate}
\item By not allowing any force imbalance in the work protocol (see
\textbf{A}1) , allows Crooks to accept Eq. (\ref{MI-Work-HamiltonianRelation})
to identify $\Delta_{\text{e}}R_{\text{C}}(\overline{\gamma}%
_{\text{C,\textsf{ab}}})$. This, according to Remark
\ref{Remark-DR-DH-relation}, implies that there is no internal microwork, a
necessity for irreversibility. From the identity in Eq.
(\ref{micro-internal-heat-workEquality}), this also means that there is no
internal microheat generated within the system. Thus, we must allow for force
imbalance in the work protocol. The force imbalance can be identified as
giving rise to an internal variable \cite{Gujrati-I}.

\item By treating the transition matrix as balanced or accepting the principle
of the detailed balance, the temperature of the system is always taken to be
$T_{0}$, the temperature of the medium. This means that the exchanged heat is
reversible. As there is no internal microheat also, there is again no
irreversibility. Therefore, one must abandon balanced transition matrices to
allow for possible heat exchange at different temperatures and make the
process irreversible. The temperature imbalance can be identified as giving
rise to an internal variable \cite{Gujrati-I}.

\item The rate of exchange of heat depends very strongly on the physical
properties such as heat conductivity, etc. of $\Sigma$ so whether the heat
transfer is isothermal or not strongly depends on how large or small is the
duration $\delta^{\prime\prime}$; we are suppressing the index $l$. This means
that it is determined by the dependence of $T_{ij}(\delta^{\prime\prime})$ on
$\delta^{\prime\prime}$. We see from Eq. (\ref{MarkovProperty2}) that%
\[
\frac{d\mathbf{p}(t^{\prime\prime})}{d\delta^{\prime\prime}}=\frac
{d\mathbf{p}(t^{\prime\prime})}{dt^{\prime\prime}}=\mathbf{p}(t^{\prime}%
)\frac{d\mathbf{T}(\delta^{\prime\prime})}{d\delta^{\prime\prime}}.
\]
Since $\mathbf{T=}\overleftrightarrow{\mathbf{T}}$ has no dependence on
$\delta^{\prime\prime}$ as shown in Claim \ref{Claim-Time-independence of T},
$d\mathbf{p}(t^{\prime\prime})/dt^{\prime\prime}$ also vanishes, which is
consistent with the conclusion in Eq. (\ref{BalancedT-form}), regardless of
$\mathbf{p}(t^{\prime})$. Therefore, the acceptance of the principle of
detailed balance does not allow for the rate of heat transfer to depend on
$\delta^{\prime\prime}$, which explains why there is no irreversibility due to
heat transfer. To describe irreversible heat transfer, we must abandon the
principle of detailed balance.
\end{enumerate}

It would be interesting to follow the consequences of abandoning the principle
of detailed balance, the lack of fluctuations of $\Delta_{\text{i}}%
S(\overline{\gamma}^{\text{(F,B)}})$, and how $\Delta Q_{\text{C}%
}^{\text{(F,B)}}(\delta_{l}^{\prime\prime})$ relates to $\Delta Q_{k}%
(\delta_{l}^{\prime\prime})$ in Eq. (\ref{Micro-heat-work}). We hope to return
to these issues in a separate publication.

Valuable communications with G. Crooks are gratefully acknowledged.


\begin{thebibliography}{99}
\bibitem {Notation}We will use the following notation in this work. We
consider a system $\Sigma$, which is in thermal contact with a heat bath (heat
medium) $\widetilde{\Sigma}_{\text{h}}$ for energy exchange and interacts with
a working medium $\widetilde{\Sigma}_{\text{w}}$ during the performance of
work done by it. The two media do not interact with each other but interact
with the system directly. We denote their combination by $\widetilde{\Sigma}=$
$\widetilde{\Sigma}_{\text{h}}\cup\widetilde{\Sigma}_{\text{w}}$ and the
combination $\Sigma\cup\widetilde{\Sigma}$ by $\Sigma_{0}$, which forms an
isolated system. All quantities pertaining to $\Sigma$\ have no suffix, and
those pertaining to $\widetilde{\Sigma}$ ($\Sigma_{0}$) with a tilde (suffix
$0$). We assume that $\Sigma$ is very small compared to $\widetilde{\Sigma}$.
The medium $\widetilde{\Sigma}$ is always taken to be in equilibrium so its
fields are not affected by the presence of $\Sigma$; as these fields are the
same for the isolated system $\Sigma_{0}$, we do not use tilde and denote them
by $T_{0},P_{0},F_{0},$ etc.
\bibitem {Note-SI}Any macroscopic extensive or intensive quantity such as the
energy $E$, entropy $S$, volume $\lambda$, temperature $T$, pressure $P$,
restoring force $F$, generalized work $dW$ and heat $dQ$, etc. that depend on
the system \cite{Notation} alone are called system-intrinsic (SI) quantities.
The external temperature $T_{0}$, pressure $P_{0},$ force $F_{0}$, etc. are
not SI quantities; they are medium-intrinsic (MI) quantities \cite{Notation}
that control the exchange (suffix e) quantities $d_{\text{e}}Z$ for the system
and determine the thermodynamic forces \cite{Note2} such as $T-T_{0}%
,P-P_{0},F+F_{0}$, etc. which control the system's approach to equilibrium.
The exchange work $d_{\text{e}}W=P_{0}dV+F_{0}d\lambda$, the exchange heat
$d_{\text{e}}Q=T_{0}d_{\text{e}}S$, etc. depend on the MI-quantities $P_{0}$,
$T_{0}$, etc.$~$and must be distinguished from SI-analogs $dW=PdV-Fd\lambda$
and $dQ=TdS$. The $\mathring{\mu}$NEQT is based on using the MI-quantities.
Adopting the use of the SI-quantities $dW$, $dQ$, etc. gives rise to a
microscopic thermodynamics, which we denote by $\mu$NEQT. It is
straightforward to express the force imbalance at a microstate level in the
$\mu$NEQT \cite{Note2} by considering the difference $d_{\text{i}}W\doteq
dW-d_{\text{e}}W\geq0,d_{\text{i}}Q\doteq dQ-d_{\text{e}}Q\geq0$, etc. that
are determined by thermodynamic forces \cite{Note2}. The inequalities are the
consequence of the second law.
\end{thebibliography}
\end{document}